\documentclass[aps,prx,showpacs,twocolumn,reprint,superscriptaddress]{revtex4-2}
\usepackage{amsmath,amssymb,amsthm,mathrsfs,amsfonts,dsfont,mathtools}
\usepackage{caption}
\usepackage{subcaption}
\usepackage{graphicx}
\usepackage{braket}
\usepackage{color}
\usepackage{hyperref}
\usepackage[capitalize]{cleveref}
\usepackage{enumitem}

\newtheorem{remark}{Remark}

\newtheorem{theorem}{Theorem}
\newtheorem{lemma}{Lemma}
\newtheorem{corollary}{Corollary}

\newtheorem{proposition}{Proposition}
\Crefname{theorem}{Theorem}{Theorems}
\theoremstyle{remark}

\begin{document}

\title{Spectral analysis of three-state quantum walks with general coin matrices}
\author{Chusei Kiumi}
\email{c.kiumi.qiqb@osaka-u.ac.jp}
\affiliation{Center for Quantum Information and Quantum Biology, Osaka University, 1-2 Machikaneyama, Toyonaka 560-0043, Japan}
\author{Jir\^o Akahori}
\email{akahori@se.ritsumei.ac.jp}
\affiliation{College of Science and Engineering Ritsumeikan University 1-1-1 Noji-higashi, Kusatsu, 525-8577, Japan}
\author{Takuya Watanabe}
\email{t-watana@se.ritsumei.ac.jp}
\affiliation{College of Science and Engineering Ritsumeikan University 1-1-1 Noji-higashi, Kusatsu, 525-8577, Japan}
\author{Norio Konno}
\email{n-konno@fc.ritsumei.ac.jp}
\affiliation{College of Science and Engineering Ritsumeikan University 1-1-1 Noji-higashi, Kusatsu, 525-8577, Japan}

\begin{abstract}
    Mathematical analysis of the spectral properties of the time evolution operator in quantum walks is essential for understanding key dynamical behaviors such as localization and long-term evolution. The inhomogeneous three-state case, in particular, poses substantial analytical challenges due to its higher internal degrees of freedom and the absence of translational invariance. We develop a general framework for the spectral analysis of three-state quantum walks on the one-dimensional lattice with arbitrary time evolution operators. Our approach is based on a transfer matrix formulation that reduces the infinite-dimensional eigenvalue problem to a tractable system of two-dimensional recursions, enabling exact characterization of eigenstates. This framework applies broadly to space-inhomogeneous models, including those with finite defects and two-phase structures. We rigorously derive necessary and sufficient conditions for the existence of point spectrum, along with a complete description of the corresponding eigenvalues and eigenstates, which are known to underlie quantum localization phenomena. Furthermore, we give a complete spectral decomposition---discrete spectrum, flat-band eigenvalues (of infinite multiplicity), and absolutely continuous spectrum---with explicit characterization of each component. Using this method, we perform exact numerical analyses of the Fourier walk with spatial inhomogeneity, revealing the emergence of localization despite its delocalized nature in the homogeneous case. Our results provide mathematical tools and physical insights into the structure of quantum walks, offering a systematic path for identifying and characterizing localized quantum states in complex quantum systems.
\end{abstract}

\maketitle

\section{Introduction}
Quantum walks (QWs) are fundamental frameworks for studying quantum dynamics. Despite their seemingly simple formulation, QWs exhibit a wide range of intricate quantum features, making them strikingly different from their classical counterpart, the random walk. These features enable the systematic design and control of quantum systems and provide a powerful framework for exploring various quantum phenomena and applications. In particular, these applications include the study of topological phases \cite{Kitagawa, Obuse2011}, the experimental realization of physical phenomena in controllable quantum systems \cite{Zhringer2010, Qiang2016, Xiao2017, Tang2018}, quantum algorithm design and universal quantum computation \cite{QFF, QSVT, Unified,universal}, and quantum transport phenomena including perfect state transfer and environment-assisted transfer in photosynthetic systems \cite{PST1, PST2, photosynthetic}.

Our focus is on discrete-time QWs on a one-dimensional integer lattice \cite{one-dim1,one-dim2}. In this framework, the interplay between the coin operator and the shift operator governs the system’s evolution. Also, the walker’s position spans the entire set of integers, resulting in an infinite-dimensional basis for the position space. This formulation is essential for modeling QWs on lattices where boundaries are either absent or effectively negligible, allowing for the study of bulk properties and phenomena such as localization \cite{Inui2005-fr,Kiumi2021-yg,Kiumi2022}.

A vital aspect of studying QWs is analyzing the spectral properties of the time evolution operator. This analysis provides insights into the behavior and features of the walk, thereby enhancing our understanding of the underlying dynamics \cite{Segawa2016-qu,smt,Kiumi2022}. A distinctive feature of infinite-dimensional systems, unlike their finite-dimensional counterparts, is the coexistence of continuous spectra and discrete eigenvalues, which reflect different physical behaviors. This structure allows for rigorous mathematical analysis \cite{one-dim1, smt}. Discrete eigenvalues in infinite systems typically correspond to localized states and are key to understanding topological phenomena in QWs. In topological phases, discrete eigenvalues within spectral gaps often indicate the presence of topologically protected bound states, which arise due to the bulk-edge correspondence \cite{Kitagawa, Obuse2011}. Unlike extended bulk states, these bound states remain robust under system variations as long as the topological invariant remains unchanged. Additionally, QWs can exhibit dynamically induced topological transitions, where spectral gaps close and reopen, signaling a shift in the topological phase. Thus, the study of discrete eigenvalues and continuous spectra in infinite QWs provides deep insights into topologically protected phenomena and phase transitions, revealing features akin to those in topological insulators and superconductors. A similar discussion arises in the context of quantum algorithm speedup \cite{QFF, QSVT, Unified,QSVT-kiumi}. In quantum algorithm design, the key objective is to determine how quickly the initial state evolves into the desired state. QWs used for algorithms are typically defined on finite-dimensional systems, and the speedup depends on the distinct eigenvalues and the spectral gap separating them from the rest of the spectrum \cite{Exponential1,Exponential2,Portugal-book,Szegedy}. A higher overlap between the initial and desired states with an eigenstate of a distinct eigenvalue, along with a large spectral gap, often enables faster transitions. This is crucial for achieving quantum speedup and is analogous to how strongly and quickly the QW localizes to the desired state.

Thus, in this paper, we focus on a mathematically explicit approach to analyze the (discrete) eigenvalues and eigenvectors of the time evolution operator. This rigorous analysis provides deeper insights into the spectral structure of the system and facilitates understanding of phenomena such as localization and the long-term behavior of QWs. It is well-known that the occurrence of localization is equivalent to the existence of eigenvalues of the time evolution operator. Moreover, a quantitative investigation of localization can be conducted by deriving the time-averaged limit distribution using the eigenvalues and eigenvectors \cite{Kiumi2022}. The transfer matrix method stands as a powerful tool for eigenvalue analysis and has been applied to several types of QW models by author \cite{Kiumi2021-yg, Kiumi2022-ts}, and even for non-unitary QWs \cite{Kiumi2022-nu}. Beyond identifying discrete eigenvalues, this paper also reveals the structure of the essential spectrum. This complete the spectral picture which is crucial for understanding the interplay between localization, delocalization, and spectral degeneracy. In particular, resolving the essential spectrum allows us to determine the presence or absence of spectral gaps, which govern the robustness of localized modes. This level of spectral resolution makes our analysis substantially more comprehensive than previous studies.

In this article, we consider the spectral analysis on the space-inhomogeneous three-state QWs on a one-dimensional integer lattice. Unlike the standard two-state QW, which has been extensively analyzed, the three-state QW exhibits fundamentally different spectral properties, leading to richer and more complex dynamics \cite{Inui2005-ry, Ko2016}. This distinction makes three-state QWs particularly intriguing and has led to intensive study in various articles across different contexts \cite{Inui2005-fr,Inui2005-ry,Falkner2014-bt,Stefanak2014-jh,Li2015-il,Machida2015-oa,Wang2015-cl,Xu2016-dy,Kawai2017-fn,Rajendran2018-ss,Endo2019-ie,Saha2021-of,Wang2015-oy,Falcao2021-pb, Kiumi2021-mg,yamagami}. We also specifically consider the two-phase model with a finite number of defects, which encompasses important and typical space-inhomogeneous models such as the one-defect model and the two-phase model \cite{Endo2020-or, Kiumi2021-yg}. These models have been studied extensively as they serve as crucial examples in understanding the impact of inhomogeneity on QWs and provide a rich landscape for investigating localization phenomena.  Previous studies \cite{Kiumi2022-ts} have conducted eigenvalue analysis for three-state QWs with a restricted class of coin matrices, including Grover matrices. In contrast, our work significantly broadens the scope by establishing a general framework that applies to arbitrary unitary coin matrices, demonstrating that precise eigenvalue characterization is achievable well beyond previously studied cases. In contrast to homogeneous two-state walks, the three-state model generically exhibits flat-band eigenvalues of infinite multiplicity, which form part of the essential spectrum and correspond to compactly supported modes originating from algebraic symmetries of the tail coins. Our transfer-matrix framework makes this structure explicit: by analyzing the asymptotic transfer matrices, we determine the absolutely continuous part of the spectrum and distinguish it from flat-band (infinite multiplicity) contributions and defect-induced (finite multiplicity) point spectra.

As the main theorem, we establish a necessary and sufficient condition for the eigenvalue equation of the time evolution operator in three-state QWs with general coin matrices, providing an explicit characterization of both the eigenvalues and the corresponding eigenvectors. In addition, we characterize the essential spectrum to obtain a full description of the spectral structure. However, calculating explicit eigenvalues analytically for a general coin matrix remains challenging due to the complexity of solving the associated algebraic equations. Therefore, we rely on numerical calculations to demonstrate the practicality of our method. For our numerical analysis, we focus on one-defect and two-phase Fourier walks. The Fourier walk has been extensively studied in the context of space-homogeneous models \cite{Saito2019,Narimatsu2021-1,Narimatsu2021-2}, where the coin matrices are identical at all positions. Notably, unlike the Grover walk, it has been proven that the Fourier walk does not exhibit localization in the homogeneous case \cite{Narimatsu2021-2}. However, our numerical results show that the Fourier walk exhibits localization in both the one-defect and two-phase models.

The remainder of this paper is organized as follows. In \cref{sec:2}, we define our three-state QWs with a self-loop on the integer lattice, incorporating both the coin and shift operators. After we explicitly define the concept of localization, we introduce the transfer matrix framework for a general coin matrix and establishes methods for addressing the eigenvalue problem. In \cref{sec:main-theorem}, we state our main results for the spectral analysis on two-phase QWs with finite defects. \cref{main theorem} provides a necessary and sufficient condition for solving the eigenvalue problem. Furthermore, \cref{theorem:spectrum-decomposition} establishes the complete spectral decomposition, providing a full classification of discrete, flat-band eigenvalue with infinite multiplicity, and absolutely continuous components of the spectrum. \cref{sec:4} focuses on the numerical analysis of the Fourier walk. We include visual representations of the numerical results for the eigenvalue equation and the probability distribution, illustrating the occurrence of localization and the existence of eigenvalues.

\section{Three state quantum walk and transfer matrix}
\label{sec:2}
\subsection{Definition of three-state quantum walk}
We define a three-state QW on the integer lattice $\mathbb{Z}$. Let us define the Hilbert space $\mathcal{H}$ as
\begin{equation*}
    \mathcal{H} =\ell ^{2} (\mathbb{Z} ;\mathbb{C}^{3} )=\left\{\Psi :\mathbb{Z}\rightarrow \mathbb{C}^{3} \ \middle| \ \sum _{x\in \mathbb{Z}} \| \Psi (x)\| _{\mathbb{C}^{3}}^{2} < \infty \right\} .
\end{equation*}
Here, $\mathbb{C}$ denotes the set of complex numbers. The quantum state $\Psi \in \mathcal{H}$ is represented as
\begin{equation*}
    \Psi (x)=\left[\begin{array}{ c }
            \Psi _{1} (x) \\
            \Psi _{2} (x) \\
            \Psi _{3} (x)
        \end{array}\right] .
\end{equation*}
Let $\{\mathsf{C}_{x}\}_{x\in \mathbb{Z} \cup \{\pm \infty \}}$ be a sequence of $3\times 3$ unitary matrices, which is given by:
\begin{equation*}
    \mathsf{C}_{x} =\left[\begin{array}{ c c c }
            a_{x}^{( 1,1)} & a_{x}^{( 1,2)} & a_{x}^{( 1,3)} \\
            a_{x}^{( 2,1)} & a_{x}^{( 2,2)} & a_{x}^{( 2,3)} \\
            a_{x}^{( 3,1)} & a_{x}^{( 3,2)} & a_{x}^{( 3,3)}
        \end{array}\right].
\end{equation*}
When either \(\left| a_{x}^{(1,3)} \right|,\ \left| a_{x}^{(2,2)} \right|,\ \left| a_{x}^{(3,1)} \right|  \) equals $1$, the coin matrix effectively reduces to a \( 2 \times 2 \) matrix. For simplicity in later discussions, we will ignore this case, assuming instead that \( \left| a_{x}^{(1,3)} \right|,\ \left| a_{x}^{(2,2)} \right|,\ \left| a_{x}^{(3,1)} \right| \neq 1 \) for all positions \( x \in \mathbb{Z} \). However, this assumption does not significantly alter the main arguments. The coin operator is then defined using the coin matrices as:
\begin{equation*}
    (C\Psi )(x)=\mathsf{C}_{x} \Psi (x).
\end{equation*}
Next, we define a shift operator $S$ that moves $\Psi _{1} (x)$ to the left and $\Psi _{3} (x)$ to the right.
\begin{equation*}
    \ (S\Psi )(x)=\left[\begin{array}{ c }
            \Psi _{1} (x+1) \\
            \Psi _{2} (x)   \\
            \Psi _{3} (x-1)
        \end{array}\right].
\end{equation*}
Finally, the time evolution of the QW is determined by the following unitary operator
\begin{equation*}
    U=SC.
\end{equation*}
For an initial state $\Psi _{0} \in \mathcal{H}\left(\Vert \Psi _{0}\Vert _{\mathcal{H}}^{2} =1\right)$, the probability of finding the walker at position $x$ at time $t \in \mathbb{Z}_{\geq 0}$ is defined as
\begin{equation*}
    \mu _{t}^{( \Psi _{0})} (x)=\left\Vert \left( U^{t} \Psi _{0}\right) (x)\right\Vert _{\mathbb{C}^{3}}^{2},
\end{equation*}
where $\mathbb{Z}_{\geq 0}$ is the set of non-negative integers.

The QW is governed by a time evolution operator $U$, and understanding its eigenvalues is essential for exploring phenomena such as localization, where the quantum walker tends to remain confined to a particular region of the lattice. The degree of localization of the walker is highly dependent on the initial state, and it is quantified by the limit distribution, which relies heavily on the overlap between the initial state and the eigenstates of $U$ \cite{Kiumi2022}. 

We say that the QW exhibits localization if there exists a position $x_{0} \in \mathbb{Z}$ and an initial state $\Psi _{0} \in \mathcal{H}$ satisfying $\lim \sup _{t\rightarrow \infty } \mu _{t}^{( \Psi _{0})}( x_{0})>0$. It is known that the QW exhibits localization if and only if there exists an eigenvalue of $U$, that is, there exists $\lambda \in [0,2\pi )$ and $\Psi \in \mathcal{H} \setminus \{\mathbf{0} \}$ such that
\begin{equation*}
    U\Psi =e^{i\lambda } \Psi.
\end{equation*}

Let \( \sigma_p(U) \) denote the set of eigenvalues of \( U \). The set of eigenvectors corresponding to \( e^{i\lambda} \in \sigma_p(U) \) is then denoted by \( \ker(U - e^{i\lambda}) \setminus \{\mathbf{0}\} \).

\subsection{Eigenvalue analysis\label{sec:3}}
In the study of QWs on a one-dimensional lattice, the transfer matrix is a crucial tool for analyzing the spectral properties of the system. By using the transfer matrix \( T_x(\lambda) \), we reduce the analysis of the infinite-dimensional Hilbert space of the QW to examining the properties of a manageable \( 2 \times 2 \) matrix, such as its determinant and trace. This reduction allows us to efficiently investigate the eigenvalues of the time evolution operator \( U \) and gain insights into the system’s behavior without directly addressing the complexities of the infinite-dimensional space.

This approach is particularly effective for studying localization, as it enables us to identify conditions for eigenvalue existence in a simplified manner. Moreover, the transfer matrix not only determines the existence of eigenvalues but also allows for the exact computation of both eigenvalues and eigenvectors, providing a constructive and comprehensive method for eigenvalue analysis.

\subsubsection{Reducing the dimensionality}
The eigenvalue equation for the QW is given by $U \Psi = e^{i \lambda} \Psi$, where $\Psi\in\mathcal{H}$ is the eigenvector and $e^{i \lambda}$ is the eigenvalue of $U$. To facilitate the analysis, this equation is transformed into a system of linear equations that describe the relationship between the components of the quantum state at different lattice positions. The key feature of this model is that the shift operator does not affect the second component of the quantum state, allowing us to simplify the problem by reducing the dimensionality to an equivalent two-state system.

The eigenvalue equation $U\Psi =e^{i\lambda } \Psi $ can be rewritten as the following system of equations:
\begin{align*}
     & e^{i\lambda } \Psi _{1} (x-1)=a_{x}^{( 1,1)} \Psi _{1} (x)+a_{x}^{( 1,2)} \Psi _{2} (x)+a_{x}^{( 1,3)} \Psi _{3} (x), \\
     & e^{i\lambda } \Psi _{2} (x)=a_{x}^{( 2,1)} \Psi _{1} (x)+a_{x}^{( 2,2)} \Psi _{2} (x)+a_{x}^{( 2,3)} \Psi _{3} (x),   \\
     & e^{i\lambda } \Psi _{3} (x+1)=a_{x}^{( 3,1)} \Psi _{1} (x)+a_{x}^{( 3,2)} \Psi _{2} (x)+a_{x}^{( 3,3)} \Psi _{3} (x).
\end{align*}
To simplify this system, we reformulate the equations by isolating the dependence of \( \Psi_2(x) \) on \( \Psi_1(\cdot) \) and \( \Psi_3(\cdot) \). The system now becomes the followings:
\begin{align}
     & e^{i\lambda } \Psi _{1} (x-1)=A_{x} (\lambda )\Psi _{1} (x)+B_{x} (\lambda )\Psi _{3} (x), \label{first}                  \\
     & e^{i\lambda} \Psi _{3} (x+1)=C_{x} (\lambda )\Psi _{1} (x)+D_{x} (\lambda )\Psi _{3} (x), \label{second}                  \\
     & \Psi _{2} (x)=\frac{a_{x}^{( 2,1)}\Psi _{1} (x)+a_{x}^{( 2,3)}\Psi _{3} (x)}{e^{i\lambda } -a_{x}^{( 2,2)}},\label{third}
\end{align}
where
\begin{align*}
     & A_{x}( \lambda ) =a_{x}^{( 1,1)} +\frac{a_{x}^{( 1,2)} a_{x}^{( 2,1)}}{e^{i\lambda } -a_{x}^{( 2,2)}} , \\
     & B_{x}( \lambda ) =a_{x}^{( 1,3)} +\frac{a_{x}^{( 1,2)} a_{x}^{( 2,3)}}{e^{i\lambda } -a_{x}^{( 2,2)}} , \\
     & C_{x}( \lambda ) =a_{x}^{( 3,1)} +\frac{a_{x}^{( 3,2)} a_{x}^{( 2,1)}}{e^{i\lambda } -a_{x}^{( 2,2)}} , \\
     & D_{x}( \lambda ) =a_{x}^{( 3,3)} +\frac{a_{x}^{( 3,2)} a_{x}^{( 2,3)}}{e^{i\lambda } -a_{x}^{( 2,2)}} . \\
\end{align*}
This reformulation highlights that \( \Psi_2(x) \) can be expressed directly in terms of \( \Psi_1(x) \) and \( \Psi_3(x) \), while \( \Psi_1(x-1) \) and \( \Psi_3(x+1) \) are determined by linear combinations of \( \Psi_1(x) \) and \( \Psi_3(x) \). Thus, \cref{third} is automatically satisfied, meaning that \( \Psi_2(x) \) is uniquely determined if \( \Psi_1(x) \) and \( \Psi_3(x) \) are known. Thus, we can focus solely on conditions \cref{first,second} for \( \Psi_1(x) \) and \( \Psi_3(x) \). Also, note that $\Psi_2(x)$ does not affect the square summability of $\Psi$ if \cref{third} is satisfied. Therefore, we can define a bijective map $\iota$ from $\left\{\Psi \in \mathcal{H} \mid \Psi\text{ satisfies \cref{third}} \right\}$ to $\ell^2(\mathbb{Z} ;\mathbb{C}^{2} )$ by
\begin{equation}\label{eq:iota}
    (\iota \Psi )(x)=\begin{bmatrix}
        \Psi _{1} (x-1) \\
        \Psi _{3} (x)
    \end{bmatrix} .
\end{equation}
Here, the inverse of $\iota $ is given as
\begin{equation*}
    \left( \iota ^{-1}\tilde{\Psi }\right) (x)=\left[\begin{array}{ c }
            \tilde{\Psi }_{1} (x+1)                                                       \\
            E_{x} (\lambda )\tilde{\Psi }_{1} (x+1)+F_{x} (\lambda )\tilde{\Psi }_{2} (x) \\
            \tilde{\Psi }_{2} (x)
        \end{array}\right]
\end{equation*}
for $\tilde{\Psi} \in \ell^2(\mathbb{Z} ;\mathbb{C}^{2} )$, $\tilde{\Psi}(x) = \begin{bmatrix} \tilde{\Psi}_1(x) \\ \tilde{\Psi}_2(x) \end{bmatrix}$. We choose this definition of $\iota$ taking the elements of $\Psi$ diagonally is to simplify the discussion by expressing the recurrence relation between positions $x$ and $x+1$ using only the coin parameters at position $x$ (the transfer matrix introduced in the next subsection becomes much simpler). Since we have this bijective map, we can equivalently convert the eigenvalue problem for three-state system on $\mathcal{H}$ to the two-state system on $\ell ^{2} (\mathbb{Z} ;\mathbb{C}^{2} )$.
\begin{lemma}
    \label{main corollary}
    For $\lambda \in [0,2\pi )$, \( e^{i\lambda}\) is an eigenvalue of $U$ if and only if there exists $\tilde{\Psi} \in \ell ^{2} (\mathbb{Z} ;\mathbb{C}^{2} )$ such that
    \begin{align}
         & e^{i\lambda } \tilde{\Psi} _{1} (x)=A_{x} (\lambda )\tilde{\Psi} _{1} (x+1)+B_{x} (\lambda )\tilde{\Psi} _{2} (x), \label{first2}   \\
         & e^{i\lambda} \tilde{\Psi} _{2} (x+1)=C_{x} (\lambda )\tilde{\Psi} _{1} (x+1)+D_{x} (\lambda )\tilde{\Psi} _{2} (x). \label{second2}
    \end{align}
    and the corresponding eigenvector is given by $\iota ^{-1}\tilde{\Psi}$.
\end{lemma}
\subsubsection{Transfer matrix on reduced space}
From the previous subsection, we have reduced the eigenvalue problem for the three-state QW to a two-state system. We will focus on \cref{first2,second2} for $\tilde{\Psi} \in \ell ^{2} (\mathbb{Z} ;\mathbb{C}^{2} )$ to analyze the eigenvalue problem. The transfer matrix $T_x(\lambda)$ is then defined to encapsulate these relationships. It is a $2 \times 2$ matrix that connects the $\tilde{\Psi}\in  \ell ^{2} (\mathbb{Z} ;\mathbb{C}^{2} )$ at position $x$ to the state at position $x+1$. The elements of the transfer matrix are derived from the components of the coin matrix, which are unitary matrices that dictate the behavior of the QW at each position. \cref{first2,second2} can then be expressed as following with $2\times2$ matrix $T_x(\lambda)$ which we call transfer matrix:
\begin{equation*}
    \tilde{\Psi}(x+1) =T_{x} (\lambda )\tilde{\Psi}(x) .
\end{equation*}
The transfer matrix \( T_x(\lambda) \) is given by the following simplified expression, derived from the unitarity of the coin matrix. A detailed derivation is provided in Appendix \ref{app:proof_prop_condition}. When we define $\iota$ in \cref{eq:iota}, we considered \( \Psi_1(x-1) \) and \( \Psi_3(x) \) as the state information, which significantly simplify the transfer matrix for this model. This approach allows the transfer matrix to be parameterized solely by the state information at site \( x \):
\begin{widetext}
    \begin{equation}\label{eq:transfer_matrix}
        T_{x}( \lambda ) =\frac{1}{a_{x}^{( 1,1)} e^{i\lambda } -e^{i\Delta _{x}}\overline{a_{x}^{( 3,3)}}}\begin{bmatrix}
            e^{i\lambda }\left( e^{i\lambda } -a_{x}^{( 2,2)}\right)                & -a_{x}^{( 1,3)} e^{i\lambda } -e^{i\Delta _{x}}\overline{a_{x}^{( 3,1)}} \\
            a_{x}^{( 3,1)} e^{i\lambda } +e^{i\Delta _{x}}\overline{a_{x}^{( 1,3)}} & -e^{i\Delta _{x}}\left( e^{-i\lambda } -\overline{a_{x}^{( 2,2)}}\right)
        \end{bmatrix} .
    \end{equation}
\end{widetext}

Here, $e^{i\Delta _{x}}$ is a determinant of the coin matrix $\det \mathsf{C}_{x}$. In the next \cref{prop:condition}, we give the main statement for this subsection, where we convert the eigenvalue equation to a reduced transfer matrix equation. Note that, when \( A_{x}(\lambda) = 0 \), which also means \( a_{x}^{(1,1)} e^{i\lambda} = e^{i\Delta_x} \overline{a_{x}^{(3,3)}} \), we cannot construct a transfer matrix. In this case, we must treat it separately. The following proposition rewrites the statement of \cref{main corollary} using the transfer matrix. The proof of this proposition is provided in \cref{app:proof_prop_condition}, there we show that \( A_{x}(\lambda) = 0 \) implies \( D_{x}(\lambda) = 0 \), thus the \cref{first2,second2} gives the direct relation \( \tilde{\Psi}_1(x) \) to \( \tilde{\Psi}_2(x) \) and \( \tilde{\Psi}_1(x+1) \) to \( \tilde{\Psi}_2(x+1) \) for \( A_{x}(\lambda) = 0 \) case, as detailed below.

\begin{proposition}\label{prop:condition}
    For $\lambda \in [0,2\pi )$, \( e^{i\lambda}\) is an eigenvalue of $U$ if and only if there exists $\tilde{\Psi} \in \ell ^{2} (\mathbb{Z} ;\mathbb{C}^{2} )$ such that
    \vspace{1em}

    \noindent$\bullet$ For $x$ with $A_{x}(\lambda) \neq 0$:
    \begin{equation*}
        \tilde{\Psi}(x+1) =T_{x}( \lambda )\tilde{\Psi}(x).
    \end{equation*}
    \noindent$\bullet$ For $x$ with $A_{x}(\lambda) =0$: there exists $k_x, k^\prime_{x}\in\mathbb{C}$ such that
    \[
       \tilde{\Psi } (x)=k_{x}v_{x},\quad \tilde{\Psi } (x+1)=k^\prime_{x}w_{x},
    \]
    where
     \begin{equation*}
            v_{x}:=\begin{bmatrix}
                a_{x}^{( 3,3)}\overline{a_{x}^{( 3,2)}} \\
                \overline{a_{x}^{( 1,1)}} a_{x}^{( 2,1)}
            \end{bmatrix} ,\quad w_{x}:=\begin{bmatrix}
                \overline{a_{x}^{( 1,1)}} a_{x}^{( 1,2)} \\
                a_{x}^{( 3,3)}\overline{a_{x}^{( 2,3)}}
            \end{bmatrix} .
        \end{equation*}
    Then, the corresponding eigenvector is given by $\iota ^{-1}\tilde{\Psi}$.
    
\end{proposition}

This result allows us to convert the eigenvalue equation \( U\Psi = e^{i\lambda} \Psi \) into a transfer matrix equation. Now we know that, to find the eigenvector \( \Psi \) of \( U \) with eigenvalue \( e^{i\lambda} \), we need to find a function \( \tilde{\Psi}: \mathbb{Z} \rightarrow \mathbb{C}^{2} \) that satisfies the condition in \cref{prop:condition} and the square summability condition \( \sum_{x \in \mathbb{Z}} \| \tilde{\Psi}(x) \|_{\mathbb{C}^{2}}^{2} < \infty \). In the next subsection, we delve deeper into the properties of the transfer matrix, and we will discuss the square summability of \( \tilde{\Psi} \) for the two-phase model with finite defects, which gives us the main results of this article.
\section{Spectral analysis for two-phase QW with defects\label{sec:main-theorem}}
In this section, we impose a condition essential for considering eigenvalue analysis with our method. 
\begin{equation*}
    \mathsf{C}_{x} =\begin{cases}
        \mathsf{C}_{\infty } ,  & x\in [ x_{+} ,\infty ) , \\
        \mathsf{C}_{x} ,        & x\in [ x_{-} ,x_{+}) ,   \\
        \mathsf{C}_{-\infty } , & x\in ( -\infty ,x_{-}),
    \end{cases}
\end{equation*}
where $x_{+}  \geq0,\ x_{-} < 0$. This implies that the coin operator is spatially homogeneous in the regions $x \geq x_+$ and $x < x_-$, while allowing for arbitrary coin matrices in the finite interval between $x_-$ and $x_+$. We say the model with this condition as two-phase QW with finite number of defects. As a remark, we can easily replace $C_{\pm\infty}$ with periodic coin matrices by following the previous research \cite{Kiumi2022-pd}. However, we opt not to explore this generalization as it leads to notably more complex notation.

\subsection{Eigenvalue and Eigenvector\label{subsec:main-theorem}}
We analyze the asymptotic behavior of \( \tilde{\Psi}(x) \) as \( x \to \pm \infty \), which is crucial for determining the square summability of \( \tilde{\Psi} \), and consequently, of \( \Psi \). For \( \tilde{\Psi}(x) \) to be square summable, it must decay as \( x \to \pm \infty \) ensuring the convergence of \( \sum_{x \in \mathbb{Z}} \| \tilde{\Psi}(x) \|^2 \), when we can construct the transfer matrix for \( x \geq x_+ \) and \( x < x_- \). Since our model assumes the coin matrix to be homogeneous in these regions, it is effective to examine the spectral properties of the transfer matrix \( T_{\pm\infty}(\lambda) \). The transfer matrix \( T_{\pm\infty}(\lambda) \) has advantageous properties for such an analysis, most notably that the norm of its determinant remains equal to \( 1 \) for any value of \( \lambda \). This invariant determinant simplifies the spectral analysis.
\begin{lemma}\label{prop:det}
    For all $\lambda \in [ -\pi ,\pi )$ and $x\in \mathbb{Z} \cup \{\pm \infty \}$, the transfer matrix given by (\ref{eq:transfer_matrix}) satisfies
    \begin{equation*}
        \left| \det T_{x} (\lambda )\right| =1.
    \end{equation*}
\end{lemma}
The proof of this lemma is given in \cref{app:det}. This result confirms that the following two eigenvalues $\zeta _{x}^{+}$, $\zeta _{x}^{-}$ of the transfer matrix satisfies $|\zeta _{x}^{+} ||\zeta _{x}^{-} |=1.$
\begin{equation*}
    \zeta _{x}^{\pm } =\frac{\mathrm{tr} (T_{x} )\pm \sqrt{\mathrm{tr} (T_{x} )^{2} -4\det T_x}}{2}.
\end{equation*}
In this paper, we define the square root of the complex number $a=|a|e^{i\theta } ,\ \theta \in [ 0,2\pi )$ as $\sqrt{a} =\sqrt{|a|} e^{i\frac{\theta }{2}}$. Consequently, we designate $\zeta _{x}^{< }$ as one of $\zeta _{x}^{\pm }$ such that $\left| \zeta _{x}^{< }\right| \leq 1$, and $\zeta _{x}^{ >}$ such that $\left| \zeta _{x}^{ >}\right| \geq 1.$ Additionally, we introduce $\ket{v_{x}^{ >}}$ and $\ket{v_{x}^{< }}$ as normalized eigenvector of $T_{x}$ corresponding to $\zeta _{x}^{ >}$ and $\zeta _{x}^{< }$, respectively. For \( \tilde{\Psi}(x) \) to decay as \( x \to \pm \infty \) it is necessary that both \( |\zeta_{\pm\infty}^{>}|\) and \( |\zeta_{\pm\infty}^{<}|\) are not equal to $1$. We define the set \( \Lambda \) as the set of \( \lambda \) for which this condition is satisfied.
\begin{equation*}
    \Lambda :=\left\{e^{i\lambda } \in \mathbb{C} \ \middle| \ \left| \zeta _{\pm \infty }^{ >}\right| \neq \left| \zeta _{\pm \infty }^{< }\right| \right\}.
\end{equation*}
The following lemma characterises the set \( \Lambda \) in terms of the trace of the transfer matrix.
\begin{lemma}\label{prop:necessary}
    \begin{equation*}
        \Lambda =\left\{e^{i\lambda }\in\mathbb{C}\mid |\mathrm{tr} (T_{\pm \infty }(\lambda) )| >2\right\},
    \end{equation*}
    where the trace of the transfer matrix is given by
    \[\mathrm{tr} (T_{x }(\lambda) )=\frac{e^{i\lambda }\left( e^{i\lambda } -a_{x}^{( 2,2)}\right) -e^{i\Delta _{x}}\left( e^{-i\lambda } -\overline{a_{x}^{( 2,2)}}\right)}{a_{x}^{( 1,1)} e^{i\lambda } -e^{i\Delta _{x}}\overline{a_{x}^{( 3,3)}}}.\]
\end{lemma}

The proof is given in the Appendix \ref{app:proof1}. This leads us to the central result of this work, stated in the following theorem:
\begin{theorem}\label{main theorem}
    For $\lambda \in [0,2\pi )$, $e^{i\lambda }$ is an eigenvalue of $U$ if and only if there exists non-zero $\tilde{\Psi } :\mathbb{Z}\rightarrow \mathbb{C}^{2}$ such that Conditions 1, 2, and 3 are all satisfied.
            \\[1em]
        \noindent$\bullet$ Condition 1. For $x\in [ x_{-} ,x_{+})$,
        \vspace{0.2em}

        \noindent\ \ $\circ$ Case 1. $A_{x}( \lambda ) \neq 0$:
        \begin{equation*}
            \tilde{\Psi } (x+1)=T_{x}( \lambda )\tilde{\Psi } (x).
        \end{equation*}
        \noindent\ \ $\circ$ Case 2. $A_{x}( \lambda ) =0$: There exists $k_{x} ,k^\prime_{x} \in \mathbb{C}$ such that
        \begin{equation*}
            \tilde{\Psi } (x)=k_{x}v_x ,\ \tilde{\Psi } (x+1)=k^\prime_{x}w_x .
        \end{equation*}
        \\[1em]
    $\bullet$ Condition 2.
        \vspace{0.2em}

        \noindent\ \ $\circ$ Case 1. $A_{\infty }( \lambda ) \neq 0$: $\tilde{\Psi }(x_{+})=\mathbf{0} $ or
            \begin{equation*}
                |\mathrm{tr} (T_{\infty } (\lambda ))| >2,\ T_{\infty }( \lambda )\tilde{\Psi } (x_{+} )=\zeta _{\infty }^{< }\tilde{\Psi } (x_{+}).
            \end{equation*}
        \noindent\ \ $\circ$ Case 2. $A_{\infty }( \lambda ) =0$: There exists $k_{x_{+}}\in \mathbb{C}$ such that
        $\tilde{\Psi } (x_{+} )=k_{x_{+}}v_{\infty},$ and for $x >x_{+} ,$ there exists $k_x\in \mathbb{C}$ such that
        \begin{equation*}
            \tilde{\Psi } (x)=\begin{cases}
                k_xv_{\infty} & \left(\frac{a_{\infty }^{( 3,3)}}{\overline{a_{\infty }^{( 1,1)}}}\right)^{2} =\frac{a_{\infty }^{( 1,2)} a_{\infty }^{( 2,1)}}{\overline{a_{\infty }^{( 3,2)}}\overline{a_{\infty }^{( 2,3)}}} , \\
                \mathbf{0}                                                       & otherwise.
            \end{cases}
        \end{equation*}
        Here, $k_x\neq 0$ only for finitely many $x$.
        \\[1em]
        \noindent$\bullet $ Condition 3.
        \vspace{0.2em}

        \noindent\ \ $\circ$ Case 1. $A_{-\infty }( \lambda ) \neq 0$: $\tilde{\Psi }(x_{-})=\mathbf{0} $ or
            \begin{equation*}
                |\mathrm{tr} (T_{-\infty } (\lambda ))| >2 ,\ T_{-\infty }( \lambda )\tilde{\Psi } (x_{-} )=\zeta _{-\infty }^{> }\tilde{\Psi } (x_{-}).
            \end{equation*}
        \ \ $\circ$ Case 2. $A_{-\infty }( \lambda ) =0$: There exists $k^\prime_{x_{-}-1}\in \mathbb{C} ,$ such that
        $\tilde{\Psi } (x_{-} )=k^\prime_{x_{-}-1}w_{-\infty},$ and for $x< x_{-} ,$ there exists $k_x\in \mathbb{C}$, such that
        \begin{equation*}
            \ \tilde{\Psi } (x)=\begin{cases}
                k_x v_{-\infty} & \left(\frac{a_{-\infty }^{( 3,3)}}{\overline{a_{-\infty }^{( 1,1)}}}\right)^{2} =\frac{a_{-\infty }^{( 1,2)} a_{-\infty }^{( 2,1)}}{\overline{a_{-\infty }^{( 3,2)}}\overline{a_{-\infty }^{( 2,3)}}} , \\
                \mathbf{0}                                                        & otherwise.
            \end{cases}
        \end{equation*}
        Here, $k_x\neq 0$ only for finitely many $x$.
\end{theorem}

    Furthermore, we have the following result on the multiplicity of eigenvalues.
\begin{proposition}\label{prop:multiplicity}
    The multiplicity of $e^{i\lambda}$ as an eigenvalue of $U$ is described as follows.

\medskip
\noindent{\rm(i)} $\dim\ker(U-e^{i\lambda})=\infty$ if and only if one of the following holds:
\begin{align*}
&A_{\infty}(\lambda)=0,\ \left(\frac{a^{\infty}_{33}}{a^{\infty}_{11}}\right)^2=\frac{a^{\infty}_{12}a^{\infty}_{21}}{a^{\infty}_{32}a^{\infty}_{23}}
\\
&A_{-\infty}(\lambda)=0,\ \left(\frac{a^{-\infty}_{33}}{a^{-\infty}_{11}}\right)^2=\frac{a^{-\infty}_{12}a^{-\infty}_{21}}{a^{-\infty}_{32}a^{-\infty}_{23}}.
\end{align*}

\medskip
\noindent{\rm(ii)}
If the condition of \rm(i) does not hold, then the multiplicity is finite and satisfies
\[
\dim\ker(U-e^{i\lambda})\ \le\ 1+N(\lambda),
\]
where $N(\lambda)$ is the number of exceptional sites in the interval $[x_-,x_+)$:
\[
N(\lambda):=|\{x\in[x_-,x_+)\;:\;A_x(\lambda)=0\}|.
\]
Equality holds exactly when, for every exceptional site $x$, there exists an eigenstate for which the coefficients $k_x$ and $k'_x$ in Condition~1, Case~2 of \cref{main theorem} are both nonzero.
\end{proposition}

\begin{proof}
By \cref{prop:condition}, if $A_x(\lambda)\neq0$ then $\tilde\Psi$ is uniquely propagated by $\tilde\Psi(x+1)=T_x(\lambda)\tilde\Psi(x)$ and no new freedom arises. 
If $A_x(\lambda)=0$, the second case replaces this by $\tilde\Psi(x)=k_xu_x$ and $\tilde\Psi(x+1)=k^\prime_xw_x$, so each exceptional site can contribute (at most) one extra degree of freedom, which is realized exactly when an eigenvector exists with $k_x\neq0$ and $k^\prime_{x}\neq0$.

If an exceptional site occurs in a tail ($A_{\pm\infty}(\lambda)=0$) and the corresponding algebraic condition in (i) holds, then \cref{main theorem} yields infinitely many
linearly independent eigenvectors from infinitely many choices of $k_x\in\mathbb{C}$, hence $\dim\ker(U-e^{i\lambda})=\infty$. Otherwise, the only exceptional sites contribute to the multiplicity lie in the finite defect
region $[x_-,x_+)$. Each of the $N(\lambda)$ exceptional sites can contribute
at most one additional independent choice, when an eigenvector exists with $k_x\neq0$ and $k^\prime_{x}\neq0$, giving $\dim\ker(U-e^{i\lambda})\le 1+N(\lambda)$. 
\end{proof}

Let us define the set of $e^{i\lambda } \in \mathbb{C}$ where the transfer matrix cannot be constructed at some position $x$: \
\begin{equation}\label{eq:lambda0}
    \Lambda _{0} =\left\{e^{i\lambda } \in \mathbb{C} \mid \exists x\in \mathbb{Z} ,\ A_{x} (\lambda )=0\right\} .
\end{equation}
In practice, the case $\lambda \in \Lambda _{0}$ can be considered separately. To find eigenvalues for $\lambda \notin \Lambda _{0}$, we must check that $|\mathrm{tr} (T_{\pm \infty } (\lambda ))| >2$, and \
\begin{align*}
    T_{-\infty } (\lambda )\tilde{\Psi } (x_{-} )                                   & =\zeta _{-\infty }^{ >}\tilde{\Psi } (x_{-} ),                                    \\
    T_{\infty } (\lambda )\left( T_{[x_{-} ,x_{+}-1 ]}\tilde{\Psi } (x_{-} )\right) & =\zeta _{\infty }^{< }\left( T_{[x_{-} ,x_{+}-1 ]}\tilde{\Psi } (x_{-} )\right) ,
\end{align*}
where $T_{[x,y]} :=T_{y} (\lambda )\cdots T_{x+1} (\lambda )T_{x} (\lambda )$. \

Thus, the theorem can be reformulated as the problem of solving the equation $\chi (\lambda )=0$ for $\lambda \notin \Lambda _{0}$, as stated in the following corollary.
\begin{corollary}\label{cor:main}
    $e^{i\lambda } \notin \Lambda _{0}$ is an eigenvalue of $U$, i.e., $e^{i\lambda } \in \sigma _{p} (U)$, if and only if the following conditions are met:
    \begin{align*}
         & |\mathrm{tr} (T_{\pm \infty }(\lambda) )| >2,                                                 \\
         & \chi (\lambda ):=T_{[ x_{-} ,x_{+}]}\ket{v_{-\infty }^{ >}} \times \ket{v_{\infty }^{< }} =0.
    \end{align*}

    Here $\times $ denotes the cross product on $\mathbb{C}^{2}$ defined by $[ x_{1} \ x_{2}] \times [ y_{1} \ y_{2}] =x_{1} y_{2} -x_{2} y_{1}$. Additionally, the multiplicity of $e^{i\lambda } \in \sigma _{p} (U)$ is $1$ with a unique eigenvector $\iota ^{-1}\tilde{\Psi }$ (up to global phase) given by:
    \begin{equation*}
        \begin{aligned}
            \tilde{\Psi } (x)= & \begin{cases}
                                     (\zeta _{\infty }^{< } )^{x-x_{+}} T_{[ x_{-} ,x_{+}-1]}\ket{v_{-\infty }^{ >}} , & x_{+} \leq x,          \\
                                     T_{[ x_{-} ,x-1]}\ket{v_{-\infty }^{ >}} ,                                        & x\in ( x_{-} ,x_{+}) , \\
                                     (\zeta _{-\infty }^{ >} )^{x-x_{-}}\ket{v_{-\infty }^{ >}} ,                      & x\leq x_{-} .
                                 \end{cases}
        \end{aligned}
    \end{equation*}
\end{corollary}
As a remark, we note that the cross product in two dimensions is equivalent to the determinant. Thus, the condition $\chi(\lambda)=0$ can be rewritten as:$$ \det \left[ T_{[x_-, x_+]} \ket{v_{-\infty}^>}, \ \ket{v_{\infty}^<} \right] = 0. $$In the context of scattering theory, $T_{[x_-, x_+]} \ket{v_{-\infty}^>}$ represents the Jost solution propagated from $-\infty$ to $x_+$. Consequently, the vanishing of this determinant implies that the Jost solutions from $+\infty$ and $-\infty$ become linearly dependent at the spectral parameter $\lambda$.

\subsection{Essential spectrum}
We now describe the spectrum of the unitary operator $U$, denoted by $\sigma(U)$, for the two-phase three-state quantum walk with finitely many defects (including its continuous part). By standard spectral theory, the spectrum admits the disjoint decomposition
\begin{equation}\label{eq:spectrum-decomposition}
\sigma(U)=\sigma_{\mathrm{disc}}(U)\ \dot\cup\ \sigma_{\mathrm{ess}}(U),    
\end{equation}
where $\sigma_{\mathrm{disc}}(U)$ is the discrete spectrum (the set of isolated eigenvalues of finite multiplicity) and $\sigma_{\mathrm{ess}}(U)$ is the essential spectrum; here $\dot\cup$ denotes a disjoint union, so that $\sigma_{\mathrm{ess}}(U)\cap\sigma_{\mathrm{disc}}(U)=\varnothing$.

 The essential spectrum is unchanged by finite defects and is determined solely by the walk's limiting behavior at $\pm\infty$. Thus, for our two-phase QW with finite defects model, the essential spectrum of the inhomogeneous walk $U$ can be expressed in terms of the tail operators:
\begin{equation}
    \sigma_{\mathrm{ess}}(U) = \sigma(U_{-\infty}) \cup \sigma(U_{+\infty}).
    \label{eq:ess-union-tails}
\end{equation}
Here $U_{-\infty},\ U_{+\infty}$ are homogeneous walks
with constant coins $\mathsf{C}_{-\infty}$ and $\mathsf{C}_{+\infty}$, respectively. For details, we refer to the general results in \cite{Richard2017, Tanaka2021}. 

Since the tail walks $U_{\pm\infty}$ are homogeneous (translation-invariant), their spectra admit the standard fiber decomposition under the Fourier transform \cite{Ko2016}, so the continuous part is purely absolutely continuous \cite{Richard2017}. Moreover, any point spectrum of a homogeneous walk can only come from a $k$-independent (flat-band) eigenvalue of the fiber matrix, and hence has infinite multiplicity. Consequently,\begin{equation}\label{eq:tail-spectrum-decomposition}
\sigma(U_{\pm\infty})=\sigma_{\mathrm{ac}}(U_{\pm\infty})\ \cup\ \sigma^\infty_{\mathrm{ess}}(U_{\pm\infty}),
\end{equation}
where $\sigma_{\mathrm{ac}}(U_{\pm\infty})$ is absolutely continuous and $\sigma^\infty_{\mathrm{ess}}(U_{\pm\infty})$ consists (if nonempty) of eigenvalues of infinite multiplicity. Our transfer-matrix formalism captures $\sigma^\infty_{\mathrm{ess}}(U_{\pm\infty})$ via \cref{cor:homogeneous}, which explicitly derives the corresponding eigenvalues and shows that they arise only at the exceptional points $\Lambda_0$ defined in \cref{eq:lambda0}.

The absolutely continuous part can be derived using Weyl's Criterion for the essential spectrum of $U_{\pm\infty}$ with non-exceptional points $\lambda\notin\Lambda_0$. A value $e^{i\lambda}$ belongs to the essential spectrum $\sigma_{\mathrm{ess}}(U)$ if and only if there exists a sequence of unit vectors $\{\Psi_n\}_{n=1}^\infty \subset \mathcal{H}$ (a Weyl sequence) such that:
    \begin{enumerate}
        \item $\|\Psi_n\| = 1$ for all $n$,
        \item $\Psi_n$ converges weakly to 0,
        \item $\| (U - e^{i\lambda})\Psi_n \| \to 0$ as $n \to \infty$.
    \end{enumerate}
    Equivalently, for $\tilde\Psi_n:=\iota\Psi_n\in\ell^2(\mathbb Z;\mathbb C^2)$, since $\iota$ is bijective, the condition is equivalent to the existence of a sequence $\{\tilde\Psi_n\}_{n=1}^\infty$ such that
\begin{enumerate}
\item $\|\tilde\Psi_n\|=1$ for all $n$,
\item $\tilde\Psi_n$ converges weakly to 0,
\item $\|( \iota U\iota^{-1}-e^{i\lambda})\tilde\Psi_n\|\to0$ as $n\to\infty$.
\end{enumerate}
     Here, note that the eigenvalue equation $(\iota U_{\pm\infty} \iota^{-1}- e^{i\lambda})\tilde{\Psi} = 0$ is equivalent to the transfer matrix recurrence $\tilde{\Psi}(x+1) = T_{\pm\infty}(\lambda)\tilde{\Psi}(x)$.

Also, since the walk has only finitely many defects, $U$ is a finite-rank (hence trace-class) perturbation of the tail operator $U_{\pm\infty}$. Therefore, by the Kato--Rosenblum theorem (see, e.g., \cite{Kato1966,Richard2018}), the absolutely continuous spectrum satisfies
\[
\sigma_{\mathrm{ac}}(U)=\sigma_{\mathrm{ac}}(U_{-\infty})\cup\sigma_{\mathrm{ac}}(U_{+\infty}).
\]

\begin{theorem}\label{theorem:spectrum-decomposition}
    The spectrum of $U$ can be decomposed as
    \[
        \sigma(U)
        =
        \sigma_{\mathrm{disc}}(U)\,\cup\,\sigma^\infty_{\mathrm{ess}}(U)\,\cup\, \sigma_{\mathrm{ac}}(U)
    \]
    where:
    \\
    (i)\ \textbf{Discrete spectrum (finite multiplicity).}
\[
\sigma_{\mathrm{disc}}(U)
=
\sigma_{\mathrm{p}}(U)\setminus\bigl(\sigma_{\mathrm{ess}}^{\infty}(U)\cup\sigma_{\mathrm{ac}}(U)\bigr),
\]
where $\sigma_{\mathrm{p}}(U)$ is given in \cref{main theorem}, and the multiplicity of each eigenvalue is bounded by $1+N(\lambda)$ (\cref{prop:multiplicity}).
    \\
    (ii) \textbf{Eigenvalues with infinite multiplicity.}
    \[
        \sigma^\infty_{\mathrm{ess}}(U)
        =
        \sigma^\infty_{\mathrm{ess}}(U_{+\infty})\cup\sigma^\infty_{\mathrm{ess}}(U_{-\infty}),
    \]
    where
    \[
        \sigma^\infty_{\mathrm{ess}}(U_{\pm\infty})
        =
        \left\{
        e^{i\lambda}\ \middle|\ 
        A_{\pm\infty}(\lambda)=0,\ 
        \left(\frac{a^{\pm\infty}_{33}}{a^{\pm\infty}_{11}}\right)^2
        =
        \frac{a^{\pm\infty}_{12}a^{\pm\infty}_{21}}{a^{\pm\infty}_{32}a^{\pm\infty}_{23}}
        \right\}.
    \]
    \\
    (iii) \textbf{Absolutely continuous spectrum.}
    \\
    \[
        \sigma_{\mathrm{ac}}(U)
        =
        \sigma_{\mathrm{ac}}(U_{+\infty})\cup\sigma_{\mathrm{ac}}(U_{-\infty}),
    \]
    where
    \[
        \sigma_{\mathrm{ac}}(U_{\pm\infty})
        =
        \left\{
        e^{i\lambda}\ \middle|\ A_{\pm\infty}(\lambda)\neq0 ,\ |\mathrm{tr} (T_{\pm \infty }(\lambda))|\le 2
        \right\}.
    \]
\end{theorem}

\begin{proof}
    The expression of $\sigma^\infty_{\mathrm{ess}}(U_{\pm\infty})$ follows directly from \cref{cor:homogeneous} applied to the tail operators $U_{\pm\infty}$. Also, \cref{prop:multiplicity} explicitly gives the expression for $\sigma^\infty_{\mathrm{ess}}(U)$ and thus we get $\sigma^\infty_{\mathrm{ess}}(U)=\sigma^\infty_{\mathrm{ess}}(U_{+\infty})\cup\sigma^\infty_{\mathrm{ess}}(U_{-\infty})$.

    To determine the absolutely continuous spectrum of the homogeneous tails $U_{\pm\infty}$, we restrict to the non-exceptional regime $A_{\pm\infty}(\lambda)\neq0$, since $A_{\pm\infty}(\lambda)=0$ only corresponds to eigenvalues in $\sigma^{\infty}_{\mathrm{ess}}(U_{\pm\infty})$. Moreover, if $|\mathrm{tr}(T_{\pm\infty}(\lambda))|>2$, then every solution $\tilde\Psi(x)$ grows exponentially on one side, and hence one cannot construct a Weyl sequence. Hence, we assume $|\mathrm{tr}(T_{\pm\infty}(\lambda))|\le2$ and $A_{\pm\infty}(\lambda)\neq0$. Then $T_{\pm\infty}(\lambda)$ has an eigenvalue $\zeta _{\pm\infty}^{\pm }$ with $|\zeta _{\pm\infty}^{\pm }|=1$, so there exists a nontrivial solution $\tilde\Psi(x)$, which is uniformly bounded ($\|\tilde\Psi(x)\|$ is bounded by some constant for all $x$) of the recurrence $\tilde\Psi(x+1)=T_{\pm\infty}(\lambda)\tilde\Psi(x)$. Let $\chi_n$ be the indicator of the interval $\{-n,\dots,n\}$ and set $\tilde\Psi_n:=\frac{\chi_n\tilde\Psi}{\|\chi_n\tilde\Psi\|}.$ Since $\tilde\Psi_n$ fails to solve the eigenvalue equation only near the two endpoints of the interval, hence $\|(\iota U_{\pm\infty}\iota^{-1} -e^{i\lambda})\tilde{\Psi}_n\|=O(n^{-1/2})\to0$ as $n\to\infty$. Therefore, $\{\iota^{-1}\tilde{\Psi}_n\}$ is a Weyl sequence and $e^{i\lambda}\in\sigma_{\mathrm{ac}}(U_{\pm\infty})$.

    Finally, since \cref{eq:ess-union-tails} and \cref{eq:tail-spectrum-decomposition}, the essential spectrum is given by $\sigma_{\mathrm{ess}}^{\infty}(U)\cup\sigma_{\mathrm{ac}}(U)$ and the discrete spectrum can be characterized as $\sigma_{\mathrm{disc}}(U)=\sigma_{\mathrm{p}}(U)\setminus\bigl(\sigma_{\mathrm{ess}}^{\infty}(U)\cup\sigma_{\mathrm{ac}}(U)\bigr)$ from \cref{eq:spectrum-decomposition}. 
\end{proof}
\begin{remark}
    Note that the above decomposition in \cref{theorem:spectrum-decomposition} does not rule out eigenvalues embedded in the absolutely continuous band set, i.e.\ phases $e^{i\lambda}\in\sigma_{\mathrm p}(U)\cap\sigma_{\mathrm{ac}}(U)$.
\end{remark}
 Such embedded eigenvalues can be created by the finite defect region when one tail is propagating at $e^{i\lambda}$ and an additional degeneracy occurs at the defects ($A_x(\lambda)=0$ at some defect site), so that an exponentially decaying mode on one side can be matched across the defect region to a propagating mode on the other, producing an $\ell^2$ eigenstate at a band energy. Any embedded eigenvalues arising in this way have finite multiplicity. In contrast, if one homogeneous tail has a flat-band eigenvalue at $e^{i\lambda}$ while the other tail is propagating at the same phase, then $e^{i\lambda}$ is an embedded eigenvalue of infinite multiplicity. A complete characterization of the parameter regimes producing embedded eigenvalues is left for future work.

\subsection{Example: Homogeneous, one-defect, two-phase model}
We will consider some simple models—homogeneous, one-defect, and two-phase QWs—that have been extensively studied, as examples to illustrate the main results.
\subsubsection{Homogeneous model}
Let us consider the simplest model where the coin matrices are homogeneous on all position $x\in \mathbb{Z}$. The coin matrix is denoted as
\begin{equation*}
    \mathsf{C}_{x} =\begin{bmatrix}
        a^{(1,1)} & a^{(1,2)} & a^{(1,3)} \\
        a^{(2,1)} & a^{(2,2)} & a^{(2,3)} \\
        a^{(3,1)} & a^{(3,2)} & a^{(3,3)}
    \end{bmatrix} .
\end{equation*}
In this setting, the eigenvalue problem becomes translation-invariant, where Fourier analysis can be applied directly, significantly simplifies the analysis. The presence of a third internal degree of freedom (due to the three coin states) results in a more intricate spectral structure.

As a corollary of our main theorem, we can explicitly characterize the eigenvalues and eigenvectors in this homogeneous case. Notably, eigenvalues only arise when a specific algebraic condition on the coin matrix entries is satisfied, indicating that localization is a highly constrained phenomenon in homogeneous three-state QWs.
\begin{corollary}\label{cor:homogeneous}
    The eigenvalue of the time evolution of homogeneous model is
    \begin{equation*}
            \sigma_{p}(U)
            =
            \begin{cases}
                \left\{
                \dfrac{\overline{a^{(3,3)}}}{a^{(1,1)}}\,e^{i\Delta}
                \right\}
                 &
                \text{if }
                \displaystyle
                \left(\dfrac{a^{(3,3)}}{\overline{a^{(1,1)}}}\right)^{2}
                =
                \dfrac{a^{(1,2)} a^{(2,1)}}{\overline{a^{(3,2)}}\,\overline{a^{(2,3)}}},
                \\[1ex]
                \emptyset
                 & \text{otherwise,}
            \end{cases}
        \end{equation*} and the corresponding eigenvectors are $\iota ^{-1}\tilde{\Psi }$ where $\tilde{\Psi }$ satisfies that there exists $k_{x} \in \mathbb{C}$
    \begin{equation*}
        \ \tilde{\Psi } (x)=k_{x}\begin{bmatrix}
            a^{( 3,3)}\overline{a^{( 3,2)}} \\
            \overline{a^{( 1,1)}} a^{( 2,1)}
        \end{bmatrix}
    \end{equation*}
    for all $x\in \mathbb{Z}$. Here, $k_{x}\neq 0$ only for finitely many $x$.
    \begin{proof}
        From the main theorem, we can state that for $\lambda \in [0,2\pi )$, $e^{i\lambda } \neq \frac{\overline{a^{(3,3)}}}{a^{(1,1)}} e^{i\Delta }$, (case $A( \lambda ) \neq 0$) is an eigenvalue of $U$ if and only if $|\mathrm{tr} (T(\lambda ))| >2$ and there exists $\tilde{\Psi } :\mathbb{Z}\rightarrow \mathbb{C}^{2} \setminus \{\mathbf{0}\}$ such that
        \begin{equation*}
            \tilde{\Psi } (x)=T( \lambda )^{x}\tilde{\Psi } (0)
        \end{equation*}
        and
        \begin{equation*}
            T( \lambda )\tilde{\Psi } (0)=\zeta ^{< }\tilde{\Psi } (0)=\zeta ^{ >}\tilde{\Psi } (0)
        \end{equation*}
        for $\tilde{\Psi}$ to be exponentially decaying on both the left and the right. However, since $\zeta ^{< } \neq \zeta ^{ >} ,$ $\tilde{\Psi } (0)=\mathbf{0}$ and this yields $\tilde{\Psi } =\mathbf{0}$. Thus, only $\frac{\overline{a^{(3,3)}}}{a^{(1,1)}}e^{i\Delta }$ (case $A( \lambda ) =0$) could be an eigenvalue. \ From the main theorem, $e^{i\lambda } =\frac{\overline{a^{(3,3)}}}{a^{(1,1)}} e^{i\Delta }$ is an eigenvalue if and only if $\left(\frac{a^{( 3,3)}}{\overline{a^{( 1,1)}}}\right)^{2} =\frac{a^{( 1,2)} a^{( 2,1)}}{\overline{a^{( 3,2)}}\overline{a^{( 2,3)}}}$ and the corresponding eigenvector becomes \ $\iota ^{-1}\tilde{\Psi }$ such that for each $x\in \mathbb{Z}$, there exists $k_{x} \in \mathbb{Z}$ such that
        \begin{equation*}
            \tilde{\Psi } (x)=k_{x}\begin{bmatrix}
                a^{( 3,3)}\overline{a^{( 3,2)}} \\
                \overline{a^{( 1,1)}} a^{( 2,1)}
            \end{bmatrix}.
        \end{equation*}
    \end{proof}
\end{corollary}
This result confirms that localization does not generally occur in homogeneous three-state QWs unless specific algebraic conditions are satisfied. Even when an eigenvalue exists, the corresponding eigenstates are fundamentally different from those obtained via the transfer matrix method in the inhomogeneous case: they are not exponentially decaying but instead are compactly supported (nonzero only at finitely many positions). This observation aligns with earlier findings in the literature, such as those derived via Fourier analysis in \cite{Ko2016,Stefanak2014StabilityPointSpectrum}, and further motivates the study of inhomogeneous models, where localization, characterized by exponentially decaying eigenstates, emerges more generically.

\subsubsection{One-defect model}
The one-defect model introduces a localized perturbation to an otherwise homogeneous QW. Specifically, the coin matrix is modified at a single site, typically taken to be the origin $x = 0$, while the coin is kept uniform elsewhere. This minimal breaking of translational invariance is sufficient to drastically alter the spectral and dynamical properties of the system. The one-defect model is particularly important in the study of QWs because it serves as a simple yet rich example of how inhomogeneities lead to localization and bound states. The one-defect model is defined such that $\mathsf{C}_{+\infty } =\mathsf{C}_{-\infty } =\mathsf{C}$, and
\begin{equation*}
    \mathsf{C}_{x} =\begin{cases}
        \mathsf{C}_{0} & x=0,     \\
        \mathsf{C}     & x\neq 0,
    \end{cases}
\end{equation*}
where we denote each element of $\mathsf{C}$ as
\begin{equation*}
    \mathsf{C}=\begin{bmatrix}
        a^{(1,1)} & a^{(1,2)} & a^{(1,3)} \\
        a^{(2,1)} & a^{(2,2)} & a^{(2,3)} \\
        a^{(3,1)} & a^{(3,2)} & a^{(3,3)}
    \end{bmatrix},\ e^{i\Delta } =\det \mathsf{C}.
\end{equation*}
Also, we write $T_{\pm \infty }( \lambda ) =T( \lambda )$, \ $\zeta _{\pm \infty }^{< } =\zeta ^{< }$ and $\zeta _{\pm \infty }^{ >} =\zeta ^{ >}$ for one-defect model. From a physical standpoint, it mimics a quantum walker encountering a single scatterer or potential barrier. It is also widely used in the context of quantum search algorithms and quantum state transfer, where the defect plays the role of a target or trap site. From our main theorem, we can precisely characterize when localization occurs by identifying discrete eigenvalues of the time evolution operator $U$. These eigenvalues correspond to states that are exponentially localized around the defect. The transfer matrix formalism developed in this paper enables us to reduce the problem to verifying whether the matching conditions at the defect site (involving $T_0(\lambda), T_{\infty}(\lambda)$, and their eigenvectors) admit a normalizable solution.
\begin{corollary}\label{cor:one-defect}
    For $\lambda \in [0,2\pi )$, $e^{i\lambda }$ is an eigenvalue of $U$ if and only if there exists $\tilde{\Psi } :\mathbb{Z}\rightarrow \mathbb{C}^{2}$ such that Conditions 1 and 2 are all satisfied.
            \\[1em]
        \noindent$\bullet $ Condition 1.
        \vspace{0.2em}

        \noindent\ \ $\circ$ Case 1. $A_{0}( \lambda ) \neq 0$:
        \begin{equation*}
            \tilde{\Psi } (1)=T_{0}( \lambda )\tilde{\Psi } (0).
        \end{equation*}
        \noindent\ \ $\circ$ Case 2. $A_{0}( \lambda ) =0$: $\exists k_{1} ,k_{2} \in \mathbb{C}$,
        \begin{equation*}
            \tilde{\Psi } (0)=k_{1}\begin{bmatrix}
                a_{0}^{( 3,3)}\overline{a_{0}^{( 3,2)}} \\
                \overline{a_{0}^{( 1,1)}} a_{0}^{( 2,1)}
            \end{bmatrix} ,\ \tilde{\Psi } (1)=k_{2}\begin{bmatrix}
                \overline{a_{0}^{( 1,1)}} a_{0}^{( 1,2)} \\
                a_{0}^{( 3,3)}\overline{a_{0}^{( 2,3)}}
            \end{bmatrix} .
        \end{equation*}
        \\[1em]
    $\bullet$ Condition 2.
        \vspace{0.2em}

        \noindent\ \ $\circ $ Case 1. $A( \lambda ) \neq 0$: $|\mathrm{tr} (T_{\infty } (\lambda ))| >2$ and
        \begin{gather*}
            T( \lambda )\tilde{\Psi } (1)=\zeta ^{< }\tilde{\Psi } (1),\\
            T( \lambda )\tilde{\Psi } (0)=\zeta ^{ >}\tilde{\Psi } (0 ).
        \end{gather*}
        \noindent\ \ $\circ $ Case 2. $A( \lambda ) =0$: $\exists k_{1} ,k_{2} \in \mathbb{C} ,$
        \begin{equation*}
            \ \tilde{\Psi } (0)=k_{1}\begin{bmatrix}
                \overline{a^{( 1,1)}} a^{( 1,2)} \\
                a^{( 3,3)}\overline{a^{( 2,3)}}
            \end{bmatrix} ,\ \tilde{\Psi } (1)=k_{2}\begin{bmatrix}
                a^{( 3,3)}\overline{a^{( 3,2)}} \\
                \overline{a^{( 1,1)}} a^{( 2,1)}
            \end{bmatrix}
        \end{equation*}
        and for $x\neq 0,1$, there exists $k_x\in \mathbb{C}$ such that
        \begin{equation*}
            \tilde{\Psi } (x)=\begin{cases}
                k_x\begin{bmatrix}
                       a^{( 3,3)}\overline{a^{( 3,2)}} \\
                       \overline{a^{( 1,1)}} a^{( 2,1)}
                   \end{bmatrix}, & \left(\frac{a^{( 3,3)}}{\overline{a^{( 1,1)}}}\right)^{2} =\frac{a^{( 1,2)} a^{( 2,1)}}{\overline{a^{( 3,2)}}\overline{a^{( 2,3)}}} , \\
                \mathbf{0},                                             & otherwise.
            \end{cases}
        \end{equation*}
        Here, $k_x\neq 0$ only for finitely many $x$.
\end{corollary}

This structure allows us to both numerically and analytically determine the presence of localized eigenstates. In particular, the one-defect model admits two distinct types of eigenstates, depending on the value of the eigenphase $\lambda$. When $A_x(\lambda) \neq 0$, the corresponding eigenstates exhibit exponential decay away from the defect site, representing true localized states in the usual physical sense. On the other hand, when $A_x(\lambda) = 0$, the eigenstates, if they exist, are compactly supported, being nonzero only at a finite number of positions. Thus, both exponentially decaying and finitely supported eigenstates can arise in the one-defect model, depending on the algebraic properties of the coin matrices and the eigenvalue under consideration. This rich spectral structure demonstrates how even a single defect can induce localization in various forms, and highlights the flexibility of our transfer matrix approach in capturing both behaviors.

\begin{figure*}[t!]
    \centering
    \begin{minipage}[b]{0.45\linewidth}
        \centering
        \includegraphics[clip, width=0.9\textwidth]{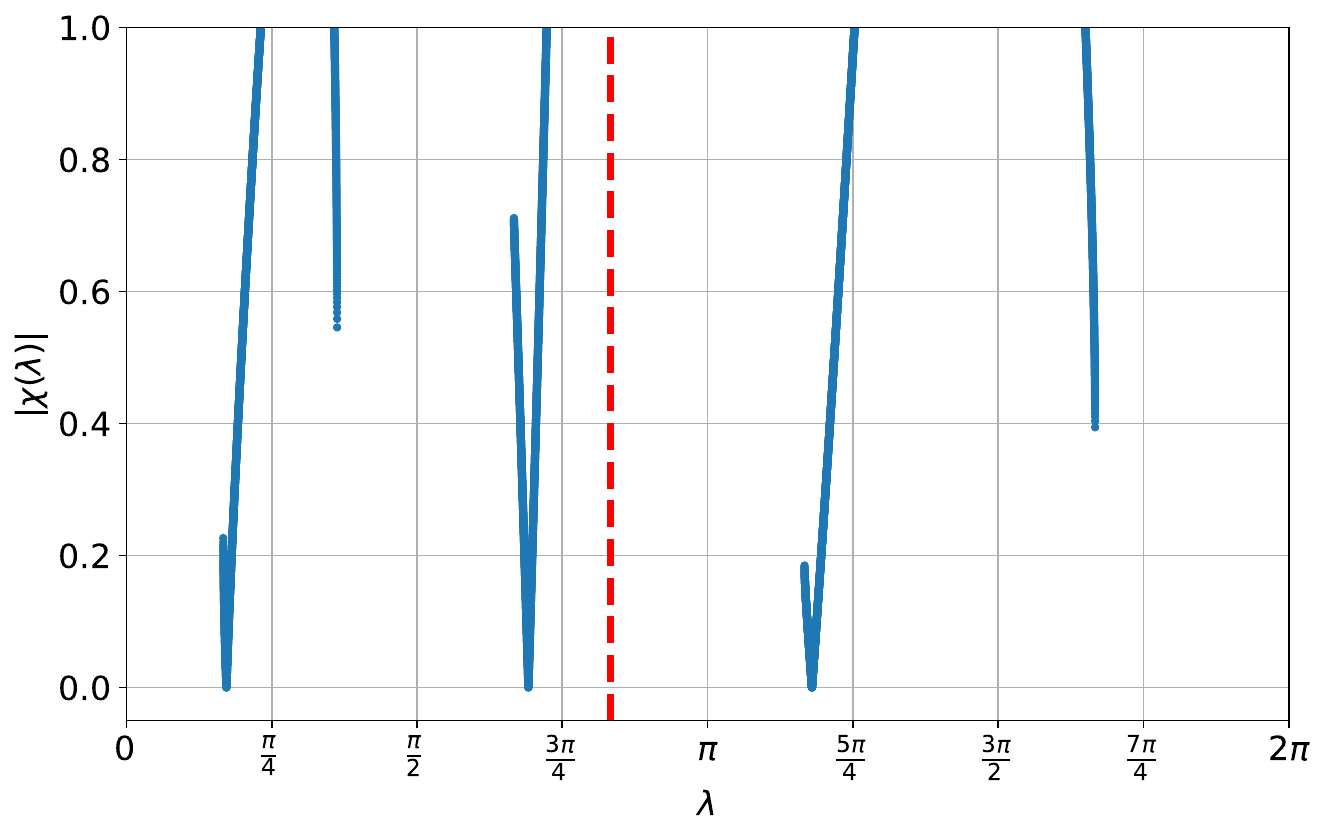}
        \subcaption{$\theta=\frac{\pi}{12}$}
    \end{minipage}
    \begin{minipage}[b]{0.45\linewidth}
        \centering
        \includegraphics[clip, width=0.9\textwidth]{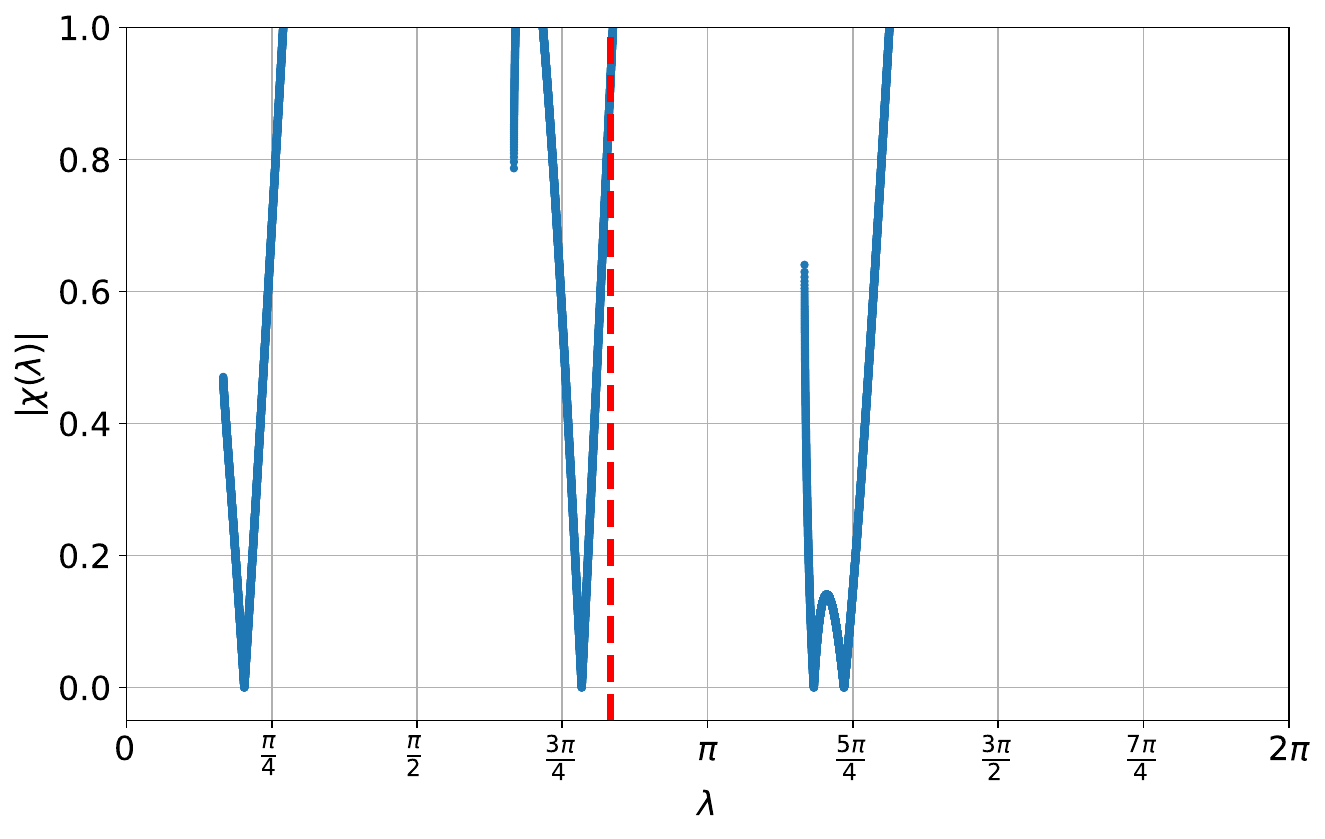}
        \subcaption{$\theta=\frac{3\pi}{12}$}
    \end{minipage}
    \begin{minipage}[b]{0.45\linewidth}
        \centering
        \includegraphics[clip, width=0.9\textwidth]{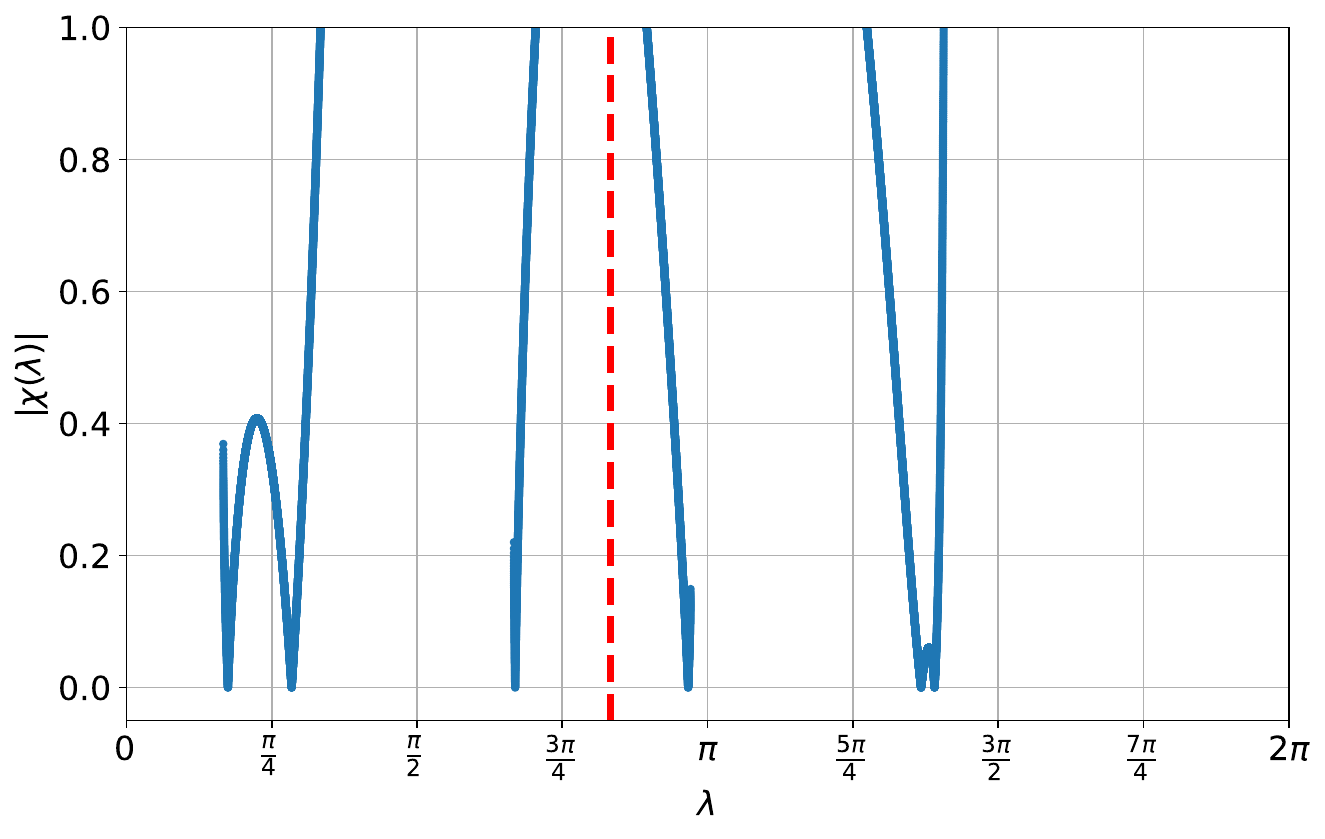}
        \subcaption{$\theta=\frac{7\pi}{12}$}
    \end{minipage}
    \begin{minipage}[b]{0.45\linewidth}
        \centering
        \includegraphics[clip, width=0.9\textwidth]{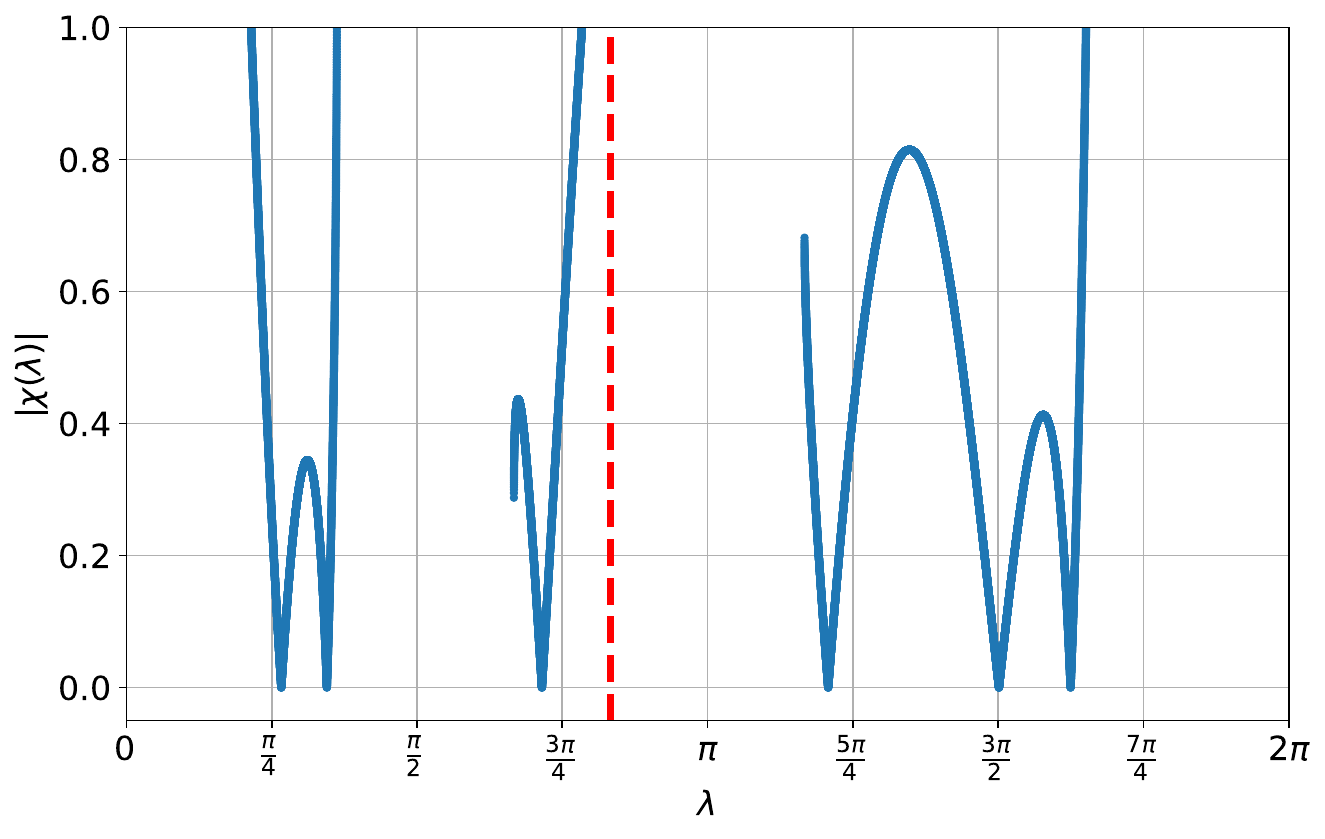}
        \subcaption{$\theta=\frac{11\pi}{12}$}
    \end{minipage}
    \captionsetup{justification=raggedright,singlelinecheck=false}
    \caption{The plot of $|\chi(\lambda)|$ from \cref{cor:main} for different values of $\theta$ in the one-defect Fourier walk. Only values of $\chi(\lambda)$ for $\lambda\in\Lambda$ are plotted, and the red line indicates $\lambda\in\Lambda_0$. By \cref{cor:main}, we see that $\lambda$ is an eigenvalue when $|\chi(\lambda)|=0$. Numerically, we can confirm that (a) has three, (b) has four, and (c) and (d) each have six eigenvalues of $U$ causing localization.}
    \label{fig:1}
\end{figure*}

\begin{figure*}[t!]
    \centering
    \begin{minipage}[b]{0.45\linewidth}
        \centering
        \includegraphics[clip, width=0.9\textwidth]{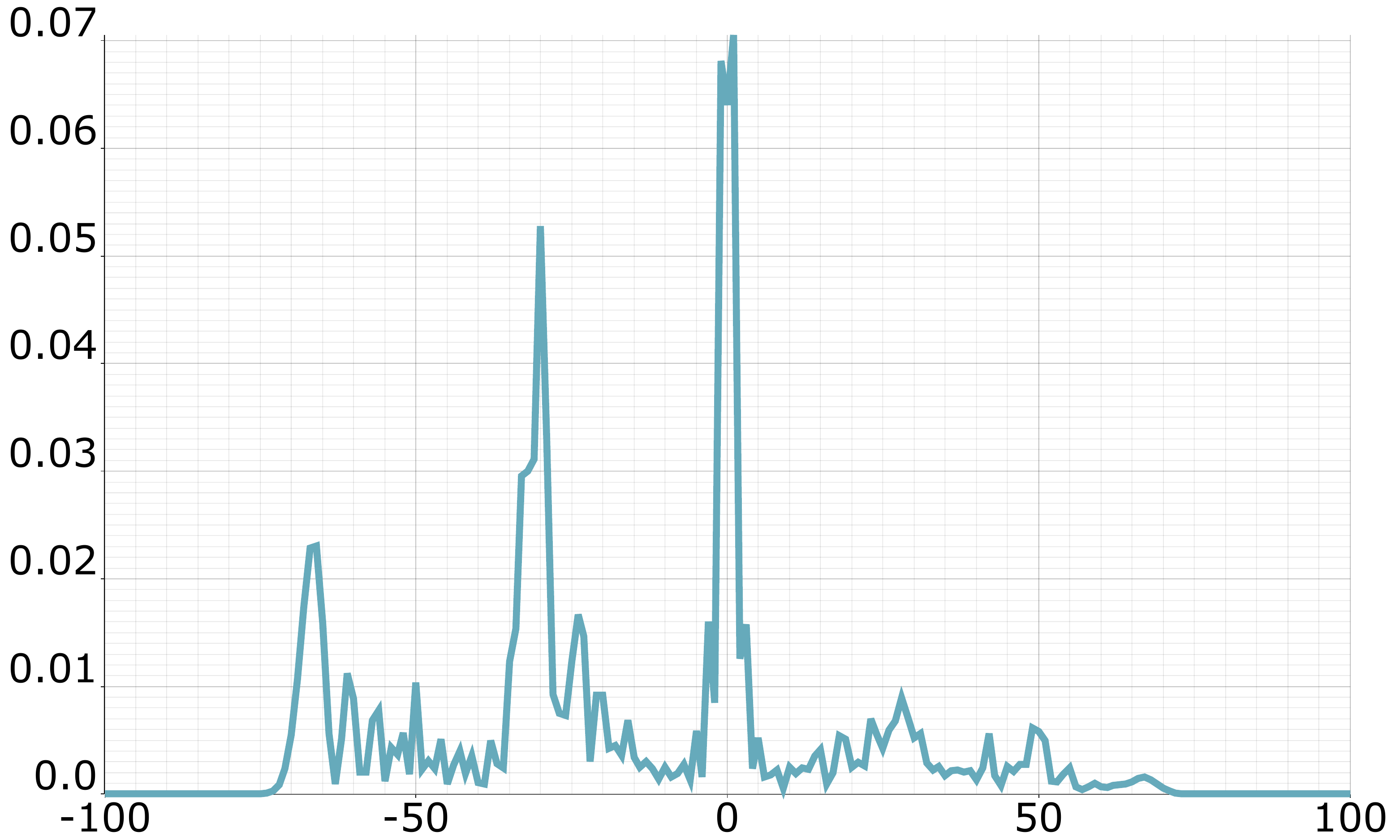}
        \subcaption{$\theta=\frac{\pi}{12}$}
    \end{minipage}
    \begin{minipage}[b]{0.45\linewidth}
        \centering
        \includegraphics[clip, width=0.9\textwidth]{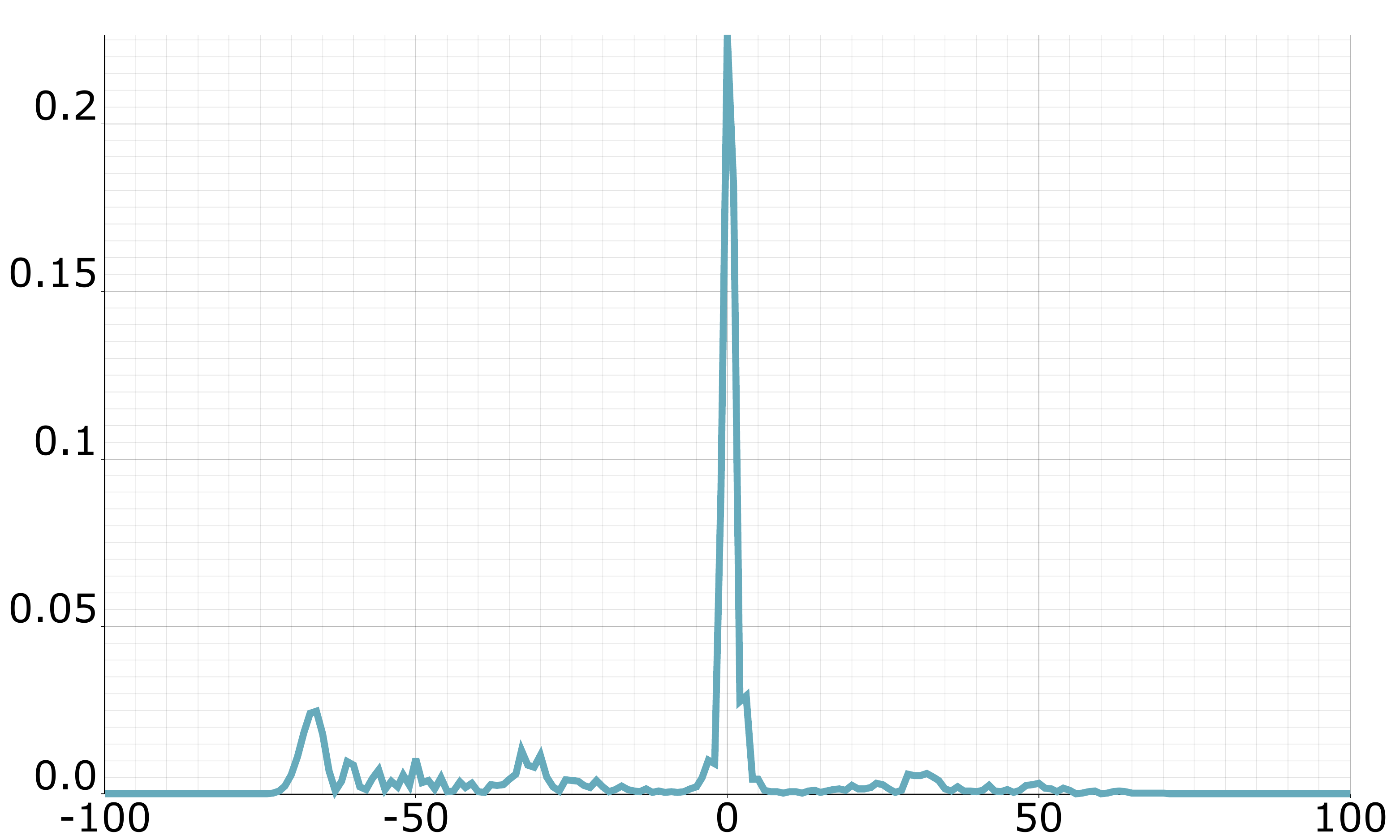}
        \subcaption{$\theta=\frac{3\pi}{12}$}
    \end{minipage}
    \begin{minipage}[b]{0.45\linewidth}
        \centering
        \includegraphics[clip,width=0.9\textwidth]{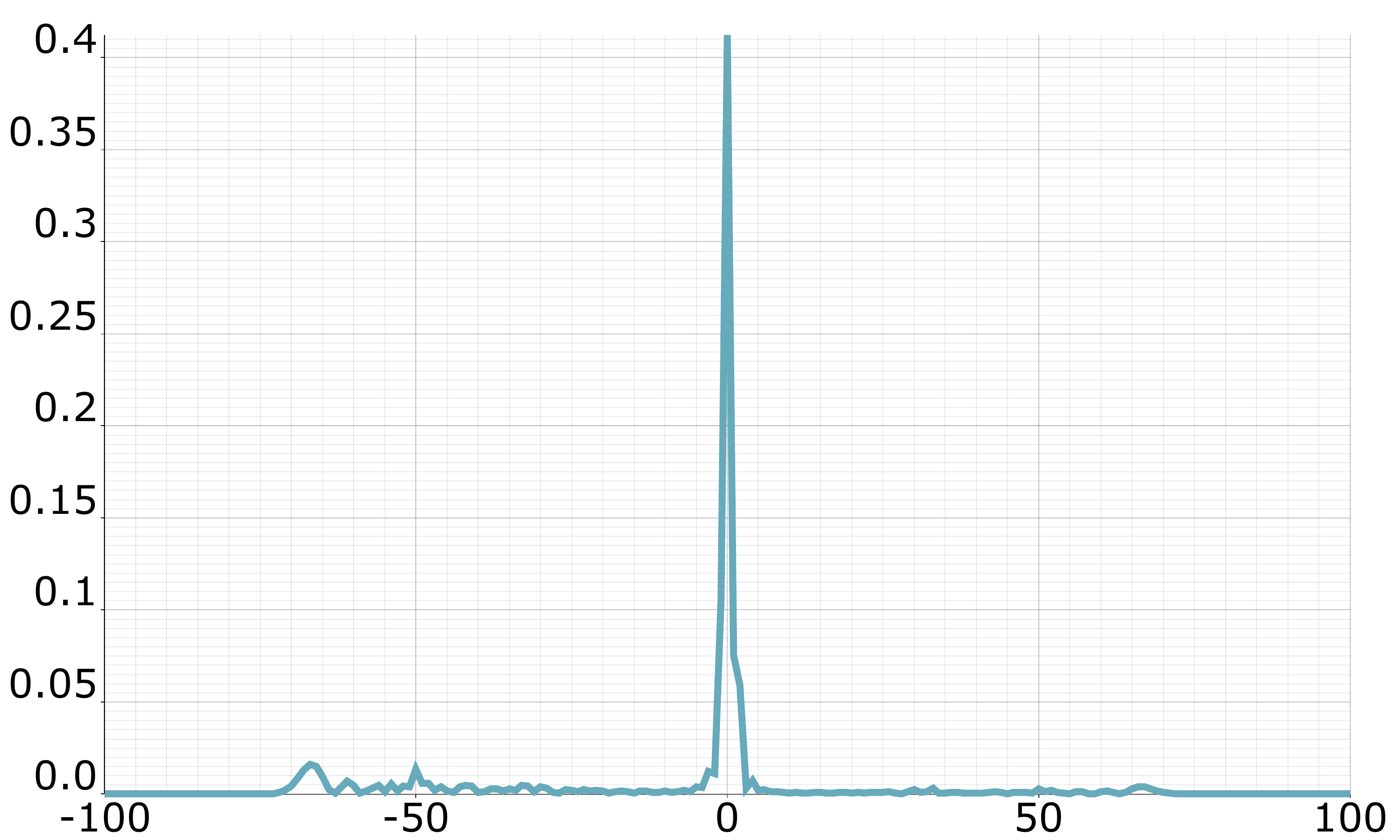}
        \subcaption{$\theta=\frac{7\pi}{12}$}
    \end{minipage}
    \begin{minipage}[b]{0.45\linewidth}
        \centering
        \includegraphics[clip, width=0.9\textwidth]{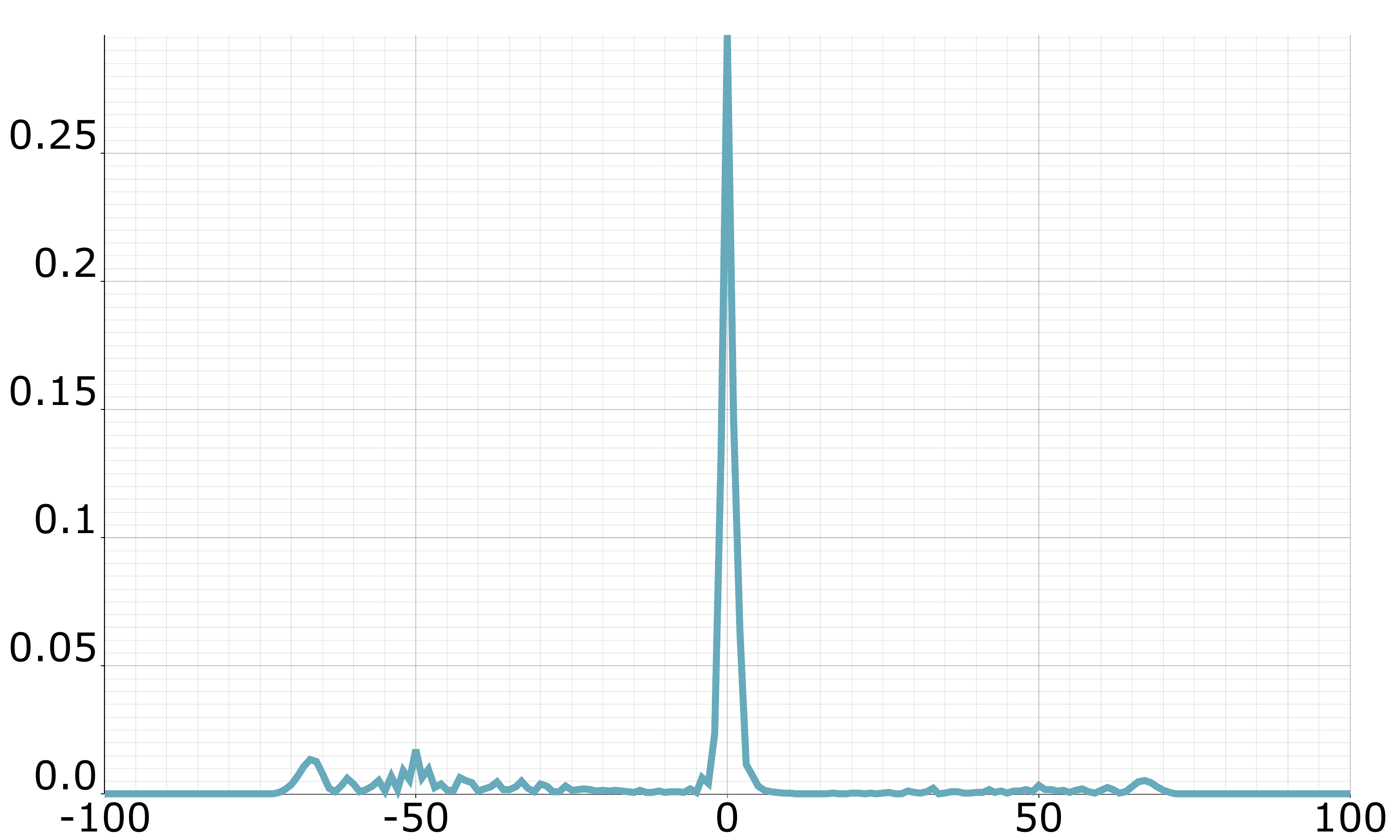}
        \subcaption{$\theta=\frac{11\pi}{12}$}
    \end{minipage}
    \captionsetup{justification=raggedright}
    \caption{Probability distribution $\mu _{t}^{( \Psi _{0})} (x)$ for the one-defect Fourier walk with different values of $\theta$ at time $100$ with initial state $\Psi_0(0)=[\frac{1}{\sqrt{3}},\frac{i}{\sqrt{3}},\frac{1}{\sqrt{3}}]^t$ and $\Psi_0(x)=\mathbf{0}$ for $x\neq 0$. All of (a-d) have a peak around the origin, which indicates localization. The result shown in \cref{fig:1} corresponds to these localization.
    \label{fig:2}}
\end{figure*}

\subsubsection{Two-phase model}
Lastly, we consider the two-phase model, where the time evolution differs between the left and right halves of the lattice. That is, the coin matrix takes two distinct values depending on whether $x < 0$ or $x \geqslant 0$. This model has been widely studied in the context of quantum localization and topological insulators, where phase boundaries can give rise to localized edge states. The two-phase model is defined such that $\mathsf{C}_{+\infty } \neq \mathsf{C}_{-\infty }$, and
\begin{equation*}
    \mathsf{C}_{x} =\begin{cases}
        \mathsf{C}_{+\infty } & x \geqslant 0, \\
        \mathsf{C}_{-\infty } & x< 0.
    \end{cases}
\end{equation*}
As a direct consequence of our main theorem, we obtain the following spectral condition for the emergence of localized modes at the interface between two distinct coin matrices. Physically, this result captures how the abrupt change in coin matrices across the origin gives rise to bound states localized near the phase boundary—analogous to edge states in topological insulators. The corollary provides a precise criterion for the existence of such localized eigenstates, which may either decay exponentially or be confined to a finite region, depending on the algebraic structure of the transfer matrices.

\begin{corollary}\label{cor:two-phase}
    For $\lambda \in [0,2\pi )$, $e^{i\lambda }$ is an eigenvalue of $U$ if and only if there exists $\tilde{\Psi } :\mathbb{Z}\rightarrow \mathbb{C}^{2}$ such that Conditions 1 and 2 are satisfied.
            \\[1em]
        \noindent$\bullet$ Condition 1.
        \vspace{0.2em}

        \noindent\ \ $\circ$ Case 1. $A_{\infty }( \lambda ) \neq 0$:
        \begin{equation*}
            T_{\infty }( \lambda )\tilde{\Psi } (0)=\zeta _{\infty }^{< }\tilde{\Psi } (0).
        \end{equation*}
        \noindent\ \ $\circ $ Case 2. $A_{\infty }( \lambda ) =0$: There exists $k\in \mathbb{C}$ such that
        \begin{equation*}
            \tilde{\Psi } (0)=k\begin{bmatrix}
                a_{\infty }^{( 3,3)}\overline{a_{\infty }^{( 3,2)}} \\
                \overline{a_{\infty }^{( 1,1)}} a_{\infty }^{( 2,1)}
            \end{bmatrix} ,
        \end{equation*}
        and for $x >0,$ there exists $k_x\in \mathbb{C}$ such that
        \begin{equation*}
            \tilde{\Psi } (x)=\begin{cases}
                k_x\begin{bmatrix}
                       a_{\infty }^{( 3,3)}\overline{a_{\infty }^{( 3,2)}} \\
                       \overline{a_{\infty }^{( 1,1)}} a_{\infty }^{( 2,1)}
                   \end{bmatrix} & \left(\frac{a_{\infty }^{( 3,3)}}{\overline{a_{\infty }^{( 1,1)}}}\right)^{2} =\frac{a_{\infty }^{( 1,2)} a_{\infty }^{( 2,1)}}{\overline{a_{\infty }^{( 3,2)}}\overline{a_{\infty }^{( 2,3)}}} , \\
                \mathbf{0}                                                        & otherwise.
            \end{cases}
        \end{equation*}
        \\[1em]
        \noindent$\bullet$ Condition 2.
        \vspace{0.2em}

        \noindent\ \ $\circ $ Case 1. $A_{-\infty }( \lambda ) \neq 0$: $|\mathrm{tr} (T_{-\infty } (\lambda ))| >2$
        \begin{equation*}
            T_{-\infty }( \lambda )\tilde{\Psi } (0)=\zeta _{-\infty }^{ >}\tilde{\Psi } (0).
        \end{equation*}
        \noindent\ \ $\circ $ Case 2. $A_{-\infty }( \lambda ) =0$: There exists $k\in \mathbb{C} ,$ such that
        \begin{equation*}
            \ \tilde{\Psi } (0)=k\begin{bmatrix}
                \overline{a_{-\infty }^{( 1,1)}} a_{-\infty }^{( 1,2)} \\
                a_{-\infty }^{( 3,3)}\overline{a_{-\infty }^{( 2,3)}}
            \end{bmatrix} ,
        \end{equation*}
        and for $x< 0,$ there exists $k_x\in \mathbb{C}$, such that
    \begin{equation*}
        \ \tilde{\Psi } (x)=\begin{cases}
            k_x\begin{bmatrix}
                   a_{-\infty }^{( 3,3)}\overline{a_{-\infty }^{( 3,2)}} \\
                   \overline{a_{-\infty }^{( 1,1)}} a_{-\infty }^{( 2,1)}
               \end{bmatrix} & \left(\frac{a_{-\infty }^{( 3,3)}}{\overline{a_{-\infty }^{( 1,1)}}}\right)^{2} =\frac{a_{-\infty }^{( 1,2)} a_{-\infty }^{( 2,1)}}{\overline{a_{-\infty }^{( 3,2)}}\overline{a_{-\infty }^{( 2,3)}}} , \\
            \mathbf{0}                                                         & otherwise.
        \end{cases}
    \end{equation*}
\end{corollary}

This characterization reveals that, like the one-defect model, the two-phase model also supports both exponentially localized eigenstates and compactly supported eigenstates. When both $A_{+\infty}(\lambda)$ and $A_{-\infty}(\lambda)$ are nonzero, the resulting eigenstates decay exponentially toward $x \to \pm\infty$, forming physically localized edge states centered at the phase boundary. In contrast, when either $A_{\pm\infty}(\lambda) = 0$, eigenstates—if they exist—must be compactly supported and are confined to a finite region. The coexistence of these two types of localized modes underscores the versatility of the two-phase model and illustrates the power of the transfer matrix formalism in capturing both spectral and spatial localization features.

\section{Numerical analysis on Fourier walk\label{sec:4}}
In this section, we demonstrate how to numerically find the eigenvalues of a three-state QW. In previous studies, the eigenvalue analysis of the generalized Grover walk has been conducted \cite{Kiumi2022-ts}. The Grover walk model exhibits a symmetry that simplifies eigenvalue analysis significantly. In contrast, this article focuses on the eigenvalue analysis of QWs with general coin matrices, allowing us to numerically determine the exact eigenvalues of non-Grover-type QWs.

To highlight the strength of our approach, we focus on the Fourier walk, a canonical example of a three-state QW with a non-Grover-type coin. The Fourier walk, which employs the discrete Fourier matrix as its coin operator, has been previously studied in this context \cite{Saito2019,Narimatsu2021-1,Narimatsu2021-2}. Due to its lack of symmetry, the Fourier walk presents greater challenges in eigenvalue analysis compared to Grover-type walks. Nevertheless, our main theorem provides a systematic framework for numerically determining eigenvalues even in such non-symmetric cases.

The coin matrix used in the Fourier walk is the discrete Fourier transform over $\mathbb{Z}_3$, defined by:
\begin{equation*}
    F=\frac{1}{\sqrt{3}}\begin{bmatrix}
        1 & 1           & 1           \\
        1 & \omega      & \omega ^{2} \\
        1 & \omega ^{2} & \omega
    \end{bmatrix}
\end{equation*}
where $\omega =e^{i\frac{2\pi }{3}}$.

Unlike the Grover coin, the Fourier coin lacks symmetry, has complex off-diagonal structure. This introduces considerable challenges for spectral analysis, particularly in detecting and classifying localized eigenstates.

Nonetheless, our main theorem offers a direct route for numerically identifying exact eigenvalues even in the absence of such simplifying structures. By reducing the eigenvalue problem to a two-component recurrence and evaluating the determinant condition associated with the composed transfer matrix, we can effectively locate discrete eigenvalues. This is particularly useful for analyzing spatially inhomogeneous configurations, such as one-defect and two-phase Fourier walks, where localization may occur.

\begin{figure*}[t]
    \begin{minipage}[b]{0.45\linewidth}
        \centering
        \includegraphics[clip, width=\textwidth]{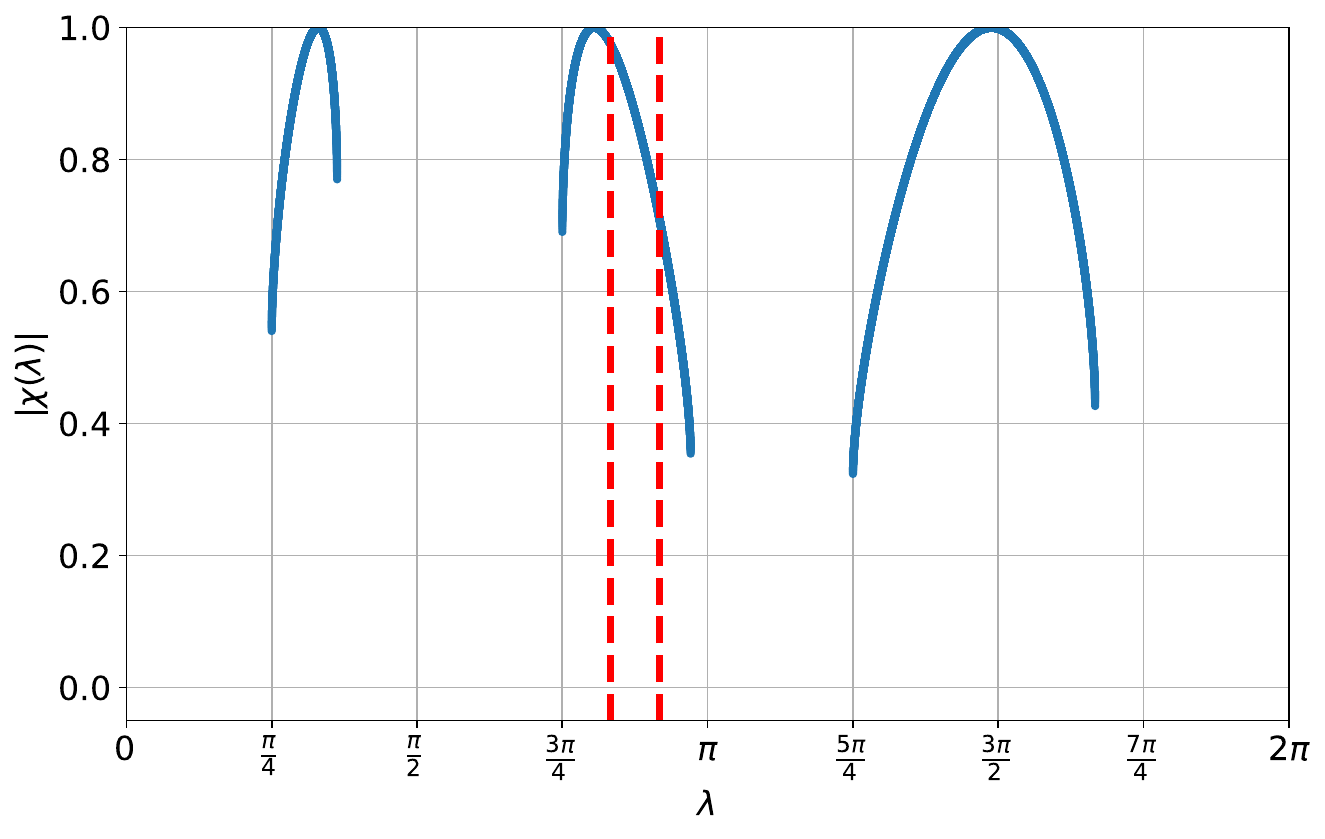}
        \subcaption{$\theta=\frac{\pi}{12}$}
    \end{minipage}
    \begin{minipage}[b]{0.45\linewidth}
        \centering
        \includegraphics[clip, width=\textwidth]{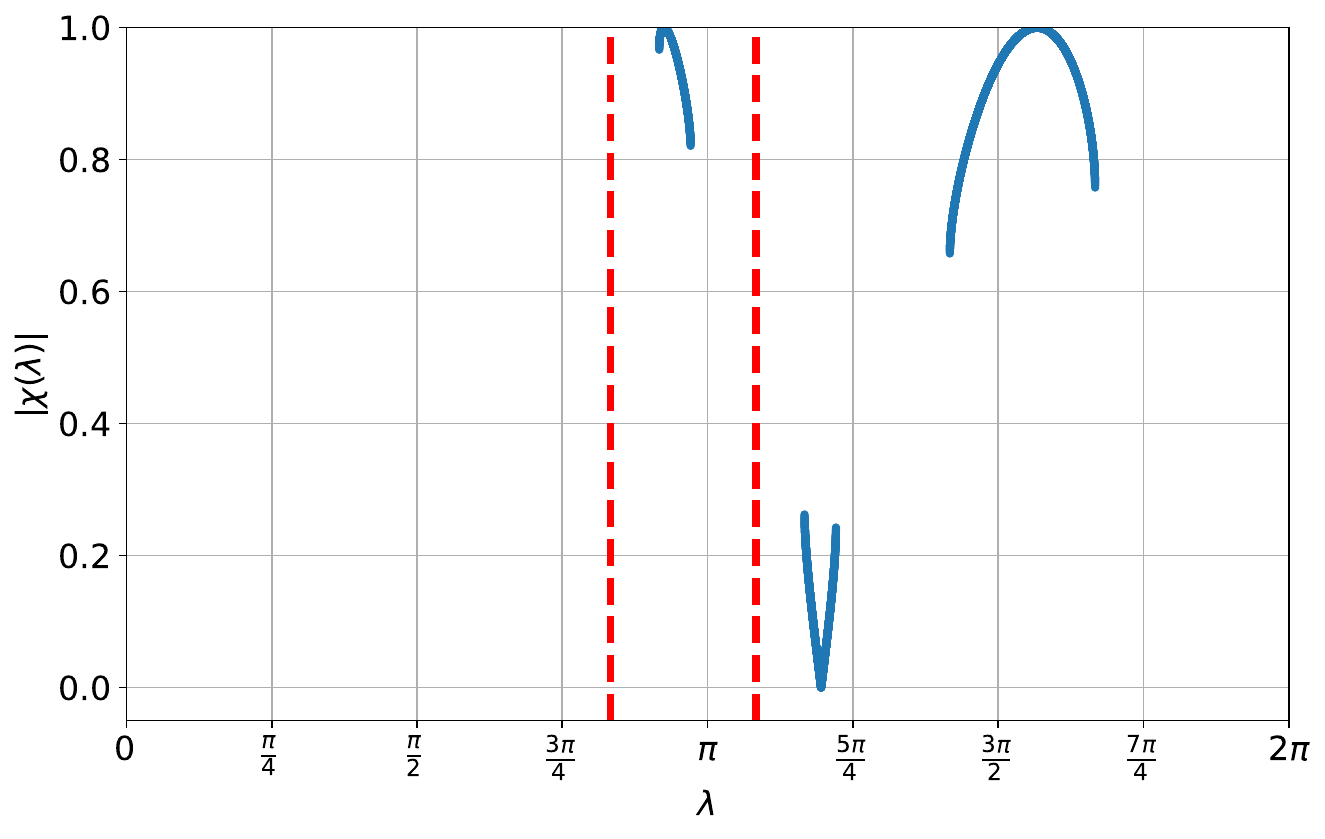}
        \subcaption{$\theta=\frac{3\pi}{12}$}
    \end{minipage}
    \begin{minipage}[b]{0.45\linewidth}
        \centering
        \includegraphics[clip, width=\textwidth]{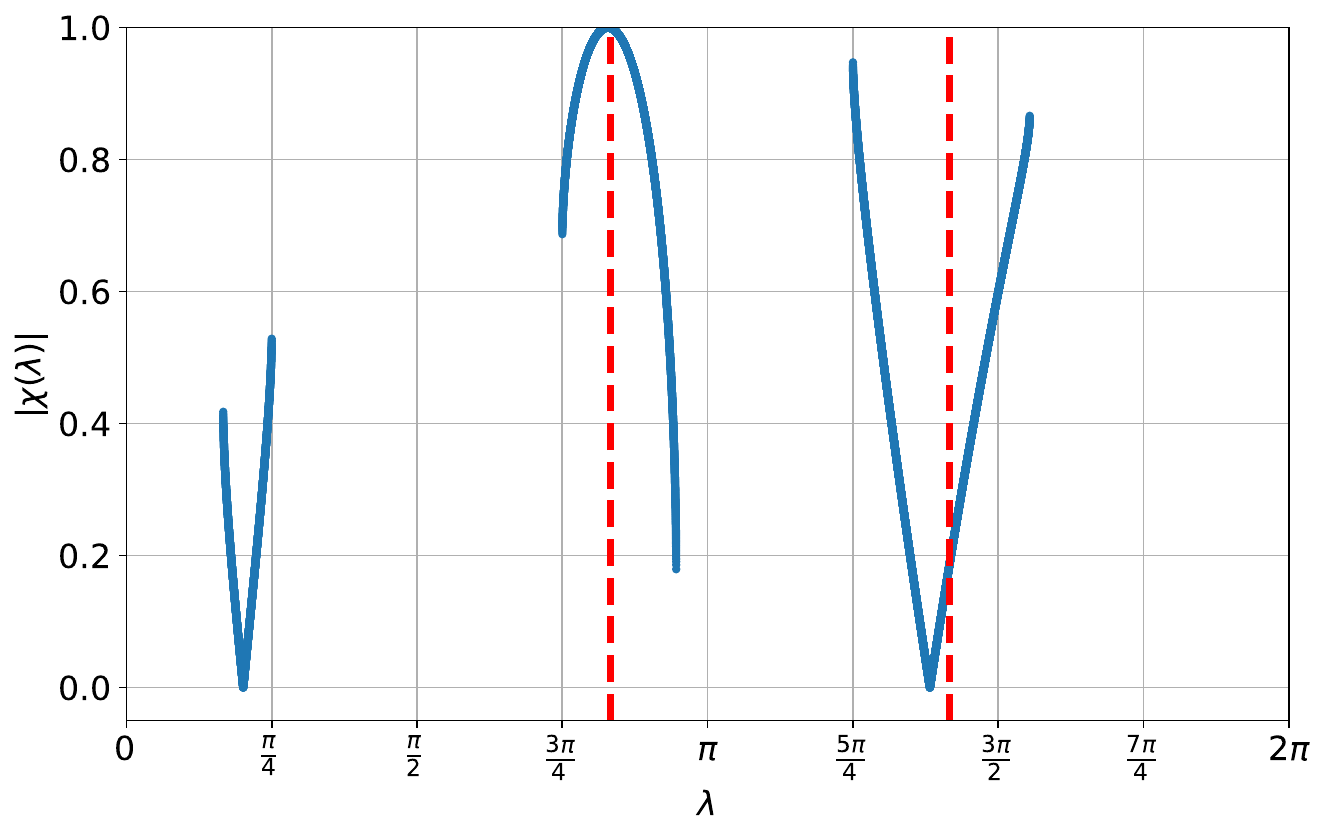}
        \subcaption{$\theta=\frac{7\pi}{12}$}
    \end{minipage}
    \begin{minipage}[b]{0.45\linewidth}
        \centering
        \includegraphics[clip, width=\textwidth]{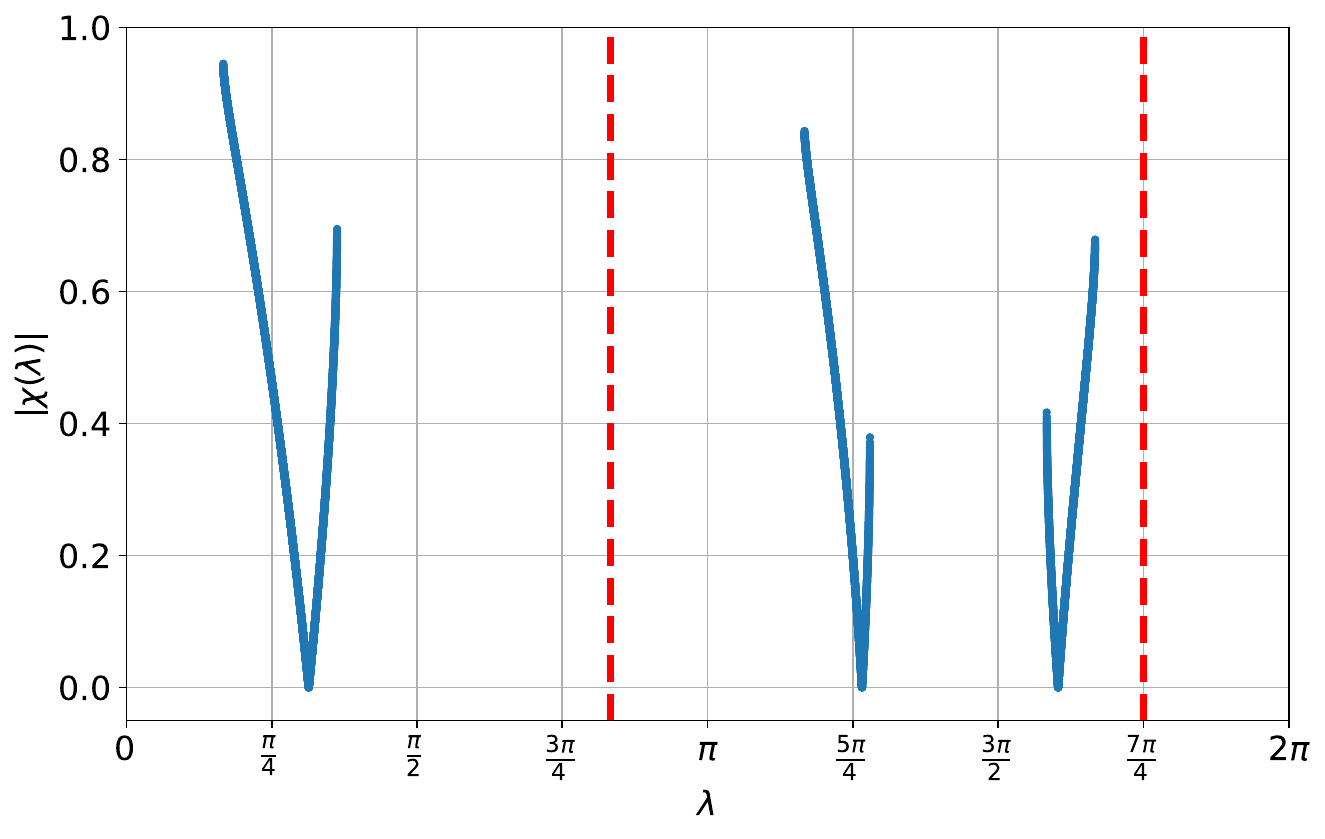}
        \subcaption{$\theta=\frac{11\pi}{12}$}
    \end{minipage}
    \captionsetup{justification=raggedright}
    \caption{
        The plot of $|\chi(\lambda)|$ from \cref{cor:main} for different values of $\theta$ in the two-phase Fourier walk. Only values of $\chi(\lambda)$ for $\lambda\in\Lambda$ are plotted, and the red line indicates $\lambda\in\Lambda_0$. By \cref{cor:main}, we see that $\lambda$ is an eigenvalue when $|\chi(\lambda)|=0$. Numerically, we can confirm that (a) has no eigenvalue, (b) has one, (c) has two, and (d) has three eigenvalues of $U$ causing localization.}
    \label{fig:3}
\end{figure*}

\begin{figure*}[t!]
    \centering
    \begin{minipage}[b]{0.45\linewidth}
        \centering
        \includegraphics[width=0.9\textwidth]{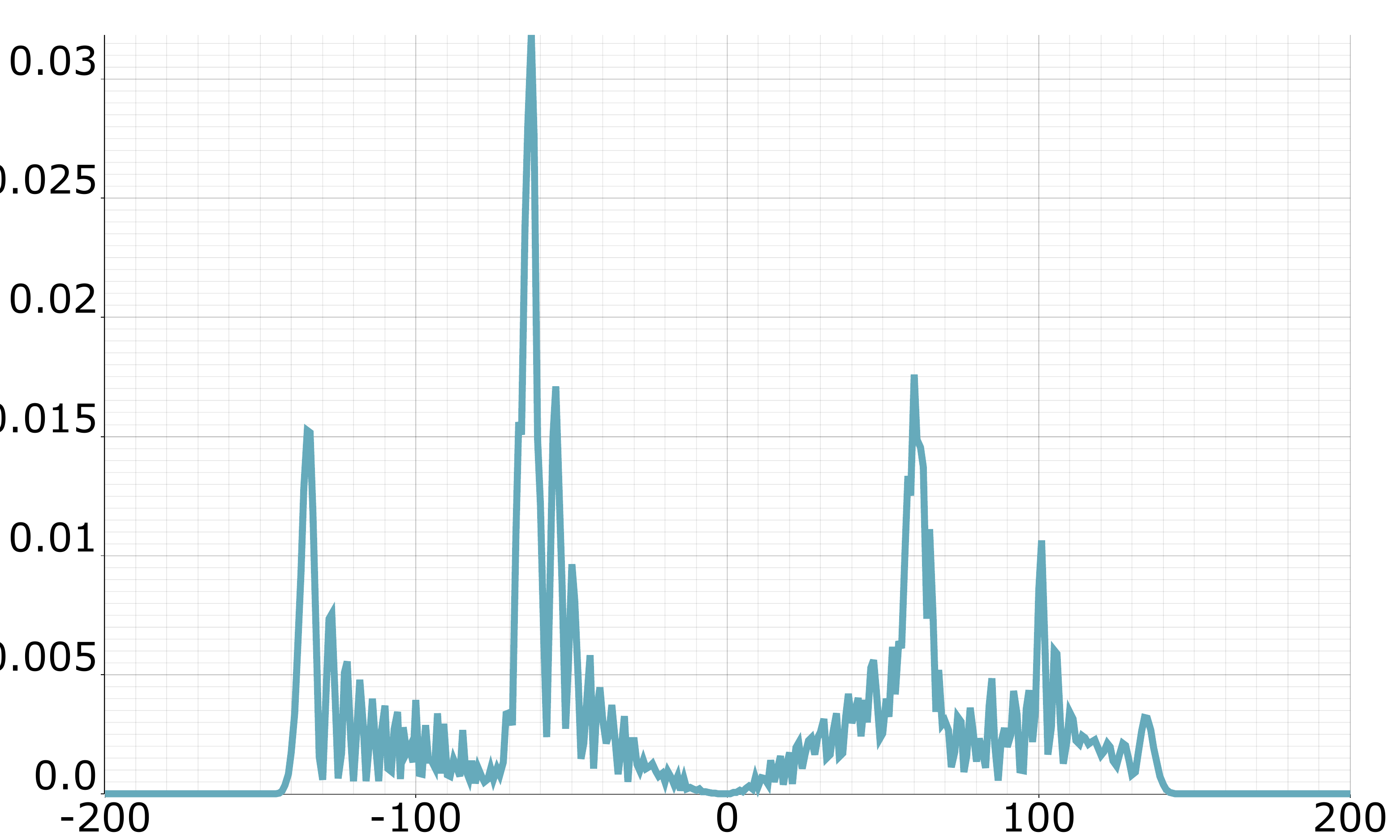}
        \subcaption{$\theta=\frac{\pi}{12}$}
    \end{minipage}
    \begin{minipage}[b]{0.45\linewidth}
        \centering
        \includegraphics[width=0.9\textwidth]{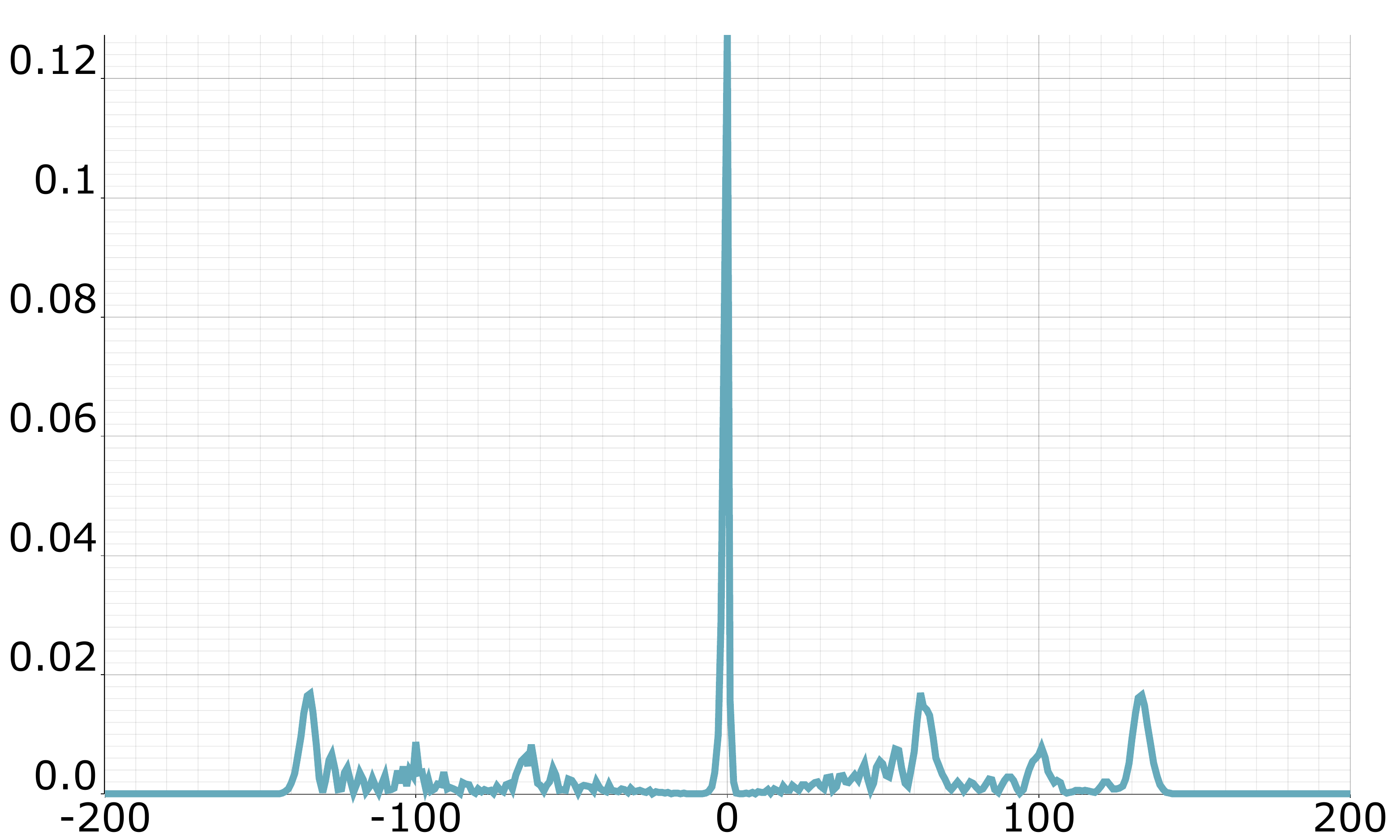}
        \subcaption{$\theta=\frac{3\pi}{12}$}
    \end{minipage}
    \begin{minipage}[b]{0.45\linewidth}
        \centering
        \includegraphics[width=0.9\textwidth]{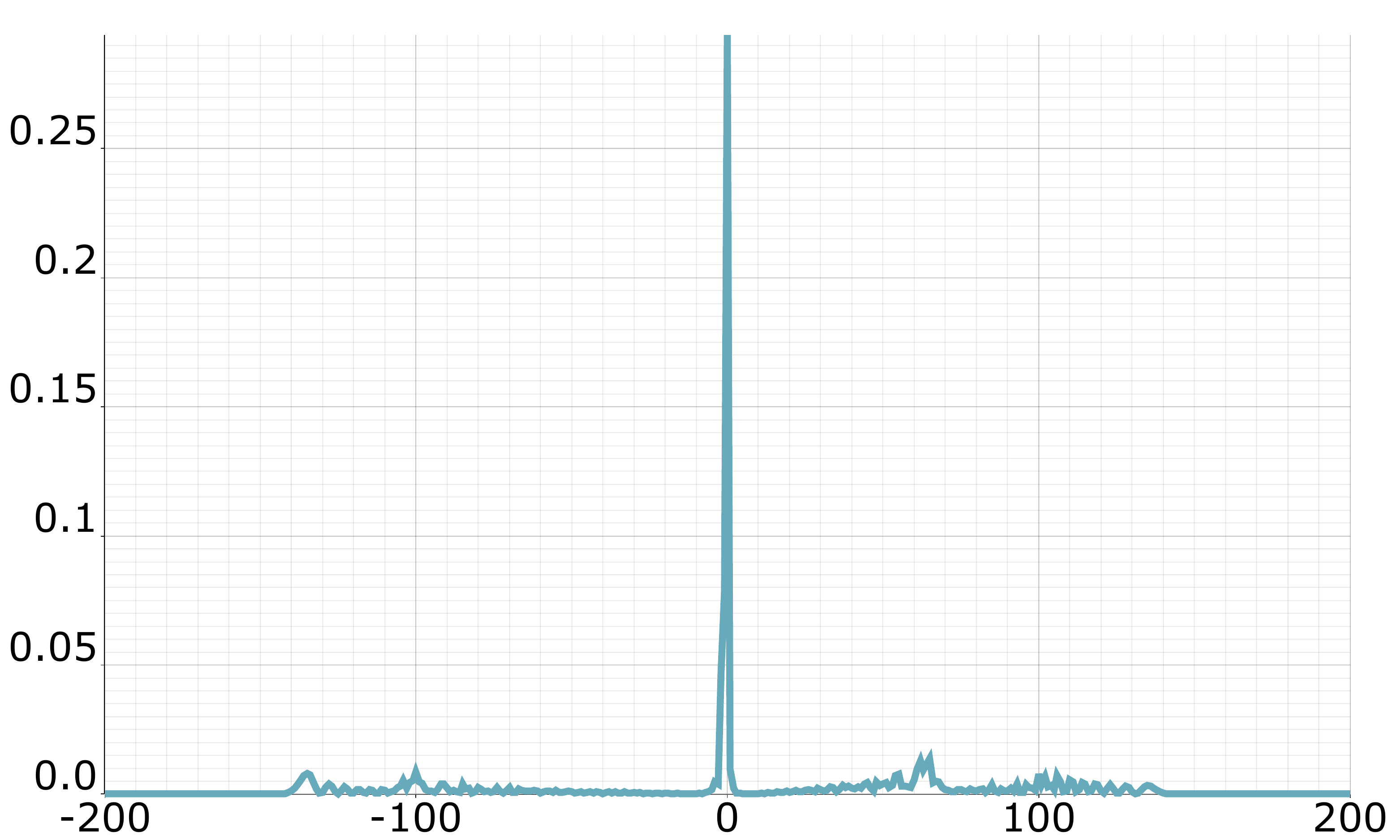}
        \subcaption{$\theta=\frac{7\pi}{12}$}
    \end{minipage}
    \begin{minipage}[b]{0.45\linewidth}
        \centering
        \includegraphics[width=0.9\textwidth]{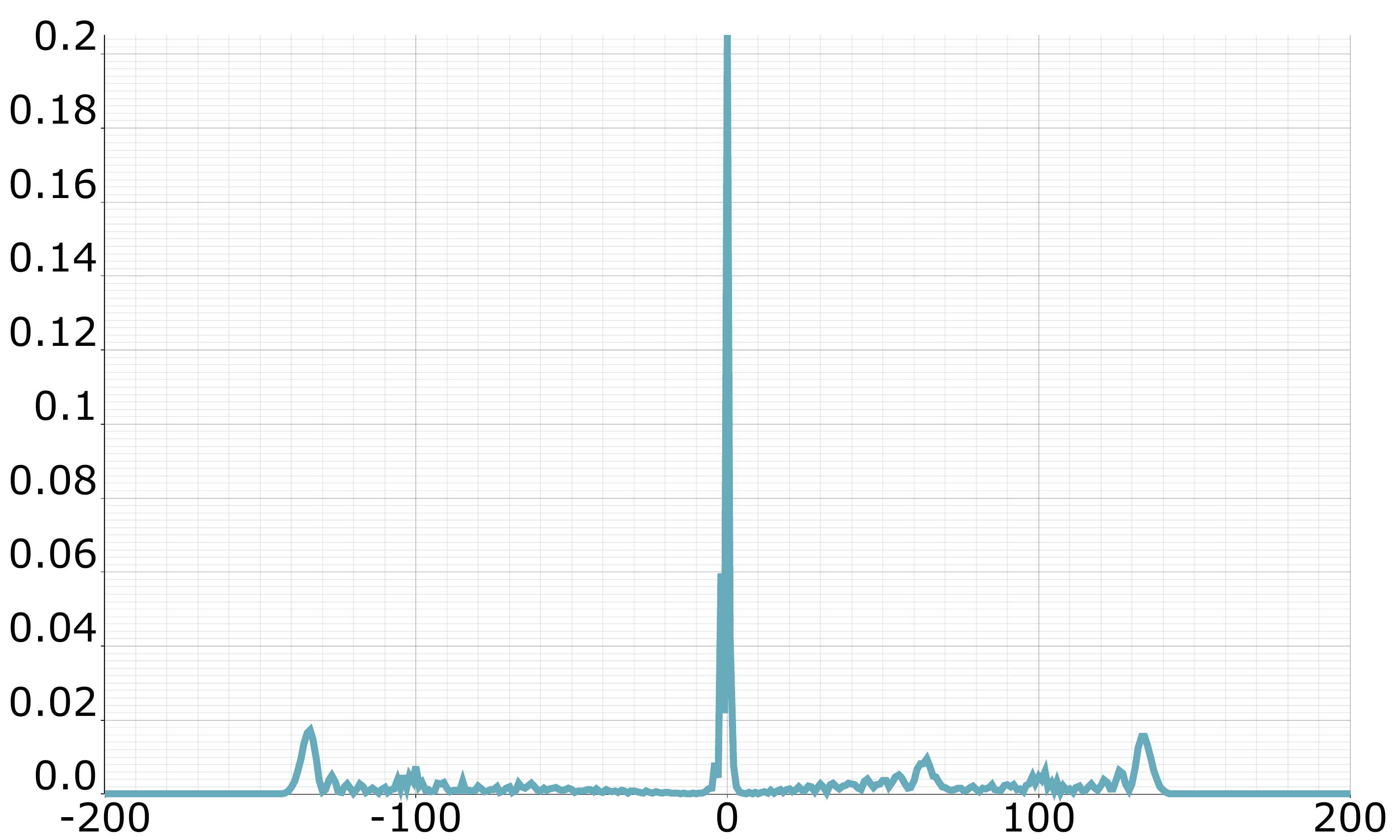}
        \subcaption{$\theta=\frac{11\pi}{12}$}
    \end{minipage}
    \captionsetup{justification=raggedright}
    \caption{Probability distribution $\mu _{t}^{( \Psi _{0})} (x)$ for the two-phase Fourier walk with different values of $\theta$ at time $100$ with initial state $\Psi_0(0)=[\frac{1}{\sqrt{3}},\frac{i}{\sqrt{3}},\frac{1}{\sqrt{3}}]^t$ and $\Psi_0(x)=\mathbf{0}$ for $x\neq 0$. All of (a-d) have a peak around the origin, which indicates localization. The result shown in \cref{fig:3} corresponds to these localization.
    }
    \label{fig:4}
\end{figure*}

\subsection{One-defect model}
We begin our numerical analysis with the one-defect Fourier walk, a simple yet revealing example of how a localized perturbation can influence the spectral structure of a three-state QW. In this model, the Fourier coin is applied uniformly across the lattice, except at the origin, where a phase $e^{i\theta } ,\ \theta \in ( 0,2\pi ]$ is introduced. The coin operator is defined as:
\begin{equation*}
    \mathsf{C}_{x} =\begin{cases}
        e^{i\theta } F & x=0,     \\
        F              & x\neq 0.
    \end{cases}
\end{equation*}
From direct computation using the structure of the Fourier matrix, we identify the set of eigenphases $\lambda$ that satisfy the algebraic condition $A_x(\lambda) = 0$, which are:
\begin{equation*}
    \Lambda _{0} =\left\{\overline{\omega } e^{-i\theta } e^{i\Delta } ,\overline{\omega } e^{i\Delta }\right\} .
\end{equation*}
These eigenphases correspond to cases where the transfer matrix loses rank and special (compactly supported) solutions could potentially exist. However, the following proposition shows that no such nontrivial eigenstates actually arise for these values.
\begin{proposition}
    \label{prop:one-defect}
    For a one-defect Fourier walk,
    \begin{equation*}
        \Lambda _{0} \cap \sigma _{p}( U) =\emptyset.
    \end{equation*}
\end{proposition}
\begin{proof}
    From \cref{cor:one-defect}, we can state that $e^{i\lambda } =\overline{\omega } e^{i\Delta }$ is not eigenvalue of $U$ since $\tilde{\Psi } (x)$ becomes $\mathbf{0}$ for all $x\in \mathbb{Z}$. Also, we prove that $\frac{\overline{a_{0}^{( 3,3)}}}{a_{0}^{( 1,1)}} e^{i\Delta _{0}}$ is not eigenvalue of $U$ by contradiction. Assume that $e^{i\lambda } =\frac{\overline{a_{0}^{( 3,3)}}}{a_{0}^{( 1,1)}} e^{i\Delta _{0}} \in \sigma _{p}( U)$ then we must have $\tilde{\Psi } :\mathbb{Z}\rightarrow \mathbb{C}^{2}$ such that for some $k_{1} ,k_{2} \in \mathbb{C}$,
    \begin{equation*}
        \tilde{\Psi } (0)=k_{1}\begin{bmatrix}
            \omega ^{2} \\
            1
        \end{bmatrix} ,\ \tilde{\Psi } (1)=k_{2}\begin{bmatrix}
            \omega \\
            1
        \end{bmatrix} .
    \end{equation*}
    and
    \begin{equation*}
        T\tilde{\Psi } (0)=\zeta ^{< }\tilde{\Psi } (0)\ \text{and} \ T\tilde{\Psi } (1)=\zeta ^{ >}\tilde{\Psi } (1).
    \end{equation*}
    However, by direct calculation, we obtain that such $\tilde{\Psi }$ does not exists other than $\mathbf{0}$.
\end{proof}
This proposition confirms that the special points in $\Lambda_0$ do not contribute to the point spectrum of the one-defect model. Consequently, our numerical search for eigenvalues can be restricted to $e^{i\lambda} \notin \Lambda_0$, where the full machinery of the transfer matrix formalism applies.

In Figure~\ref{fig:1}, we plot the function $\chi(\lambda)$ for several values of the defect phase $\theta$. The zeros of $\chi(\lambda)$ indicate the presence of eigenvalues. Figure~\ref{fig:2} presents the corresponding probability distributions for the localized states, clearly demonstrating localization around the defect site. These results show how tuning the phase at a single site can control the emergence or disappearance of localized modes.

\subsection{Two-phase model}
Next, we analyze the two-phase Fourier walk, the coin matrix differs between the left and right halves of the lattice, introducing a discontinuity at the origin that mimics a phase boundary. Such models are of significant interest in the study of topological phases, where localized edge states are expected to emerge at the interface between different quantum “phases.” Here, we consider two-phase Fourier walk defined by the coin operator, where we apply relative phase that distinguishes the two regions $e^{i\theta } ,\ \theta \in ( 0,2\pi ]$ for $x\geqslant 0$:
\begin{equation*}
    \mathsf{C}_{x} =\begin{cases}
        F              & x< 0 ,   \\
        e^{i\theta } F & x\geq 0.
    \end{cases}
\end{equation*}
As in the one-defect case, we first identify the set $\Lambda_0$ of exceptional eigenphases at which the transfer matrix loses rank:
\begin{equation*}
    \Lambda_0 = \left\{\overline{\omega } e^{i\Delta } ,\overline{\omega } e^{i\theta } e^{i\Delta }\right\}.
\end{equation*}
These values potentially allow compactly supported eigenstates, depending on whether certain algebraic conditions are met. However, as shown in the following proposition, neither of these eigenphases actually contributes to the point spectrum. The following Proposition can be proved similarly to \cref{prop:one-defect} by investigating \cref{cor:two-phase}.
\begin{proposition}
    For a two-phase Fourier walk,
    \begin{equation*}
        \Lambda _{0} \cap \sigma _{p}( U) =\emptyset
    \end{equation*}
\end{proposition}
This result simplifies the numerical task of locating eigenvalues in the two-phase model: we need only consider $e^{i\lambda} \notin \Lambda_0$, where exponentially localized solutions may exist and be captured by our transfer matrix framework.

In Figure~\ref{fig:3}, we plot the determinant function $\chi(\lambda)$ for several values of the phase parameter $\theta$. The zeros of $\chi(\lambda)$ indicate the existence of discrete eigenvalues. Figure~\ref{fig:4} displays the probability distributions corresponding to these eigenphases, revealing localized states centered at the phase boundary. These results vividly demonstrate how a simple phase jump across the origin can give rise to spatially localized eigenmodes, analogous to edge states in topological systems. The ability of our method to detect and characterize such states in a non-Grover setting underscores its generality and practical power.

\section{Conclusion and Discussion}
In this paper, we developed a general transfer matrix method for analyzing the spectral properties of space-inhomogeneous three-state quantum walks (QWs) on the integer lattice with general coin matrices. This extends previous results \cite{Kiumi2022-ts}, which were limited to specific classes such as Grover-type coins, to a broader class of models including non-Grover-type coins like the Fourier matrix. Furthermore, we explicitly characterized the essential spectrum, identifying both the flat-band eigenvalues with infinite multiplicity and the absolutely continuous spectrum. This completes the full spectral classification of the model. Our approach reduces the eigenvalue problem of the original three-state QW to a tractable two-component recurrence relation, enabling an explicit characterization of eigenvalues and eigenstates through a $2 \times 2$ transfer matrix formalism. We derived a necessary and sufficient condition for the eigenvalue problem in the case of two-phase QWs with a finite number of defects, providing a mathematical characterization and enabling a precise analysis of localization phenomena in these systems.

As a concrete application of our framework, we numerically analyzed the Fourier walk—a three-state QW with a coin matrix not amenable to earlier analytic techniques—under one-defect and two-phase settings. Despite the fact that the homogeneous Fourier walk is known to be delocalized, our results confirm that both the one-defect and two-phase modifications induce localization. We verified the existence of discrete eigenvalues via plots of the function $\chi(\lambda)$ and demonstrated localization through the time evolution of the probability distribution.

These results not only provide new insights into the spectral theory of three-state QWs but also highlight the effectiveness of our method in studying localization beyond symmetric models such as the Grover walk. Our framework is particularly relevant for exploring quantum transport phenomena, topological phases, and the stability of localized states under perturbations.
\subsection{Physical implications}
    Our transfer-matrix framework provides a complementary and rigorous method for analyzing the spectrum of three-state quantum walks. In particular, the eigenvalue conditions established in \cref{main theorem} allow for a precise characterization of localized states induced by defects or phase boundaries, while the resulting three-way spectral decomposition offers a clear physical interpretation of the dynamics generated by general three-state coins:
    \begin{itemize}
        \item Discrete eigenvalues in $\sigma_{\mathrm{disc}}(U)$, correspond to genuine defect or edge modes that decay exponentially away from the inhomogeneous region. Any initial state with nonzero overlap on such eigenvectors exhibits persistent localization near the defect, as confirmed numerically in \cref{sec:4}.
        \item Flat-band points in $\sigma^\infty_{\mathrm{ess}}(U)$ arise from special symmetries of the tail coins and support massively degenerate, compactly supported states. These modes are spatially localized but form part of the essential spectrum due to their infinite multiplicity.
        \item The absolutely continuous part $\sigma_{\mathrm{ac}}(U)$ corresponds to extended states. For initial conditions overlapping only with this part, the walk exhibits ballistic spreading and delocalization.
    \end{itemize}

    Altogether, the transfer-matrix approach not only pinpoints the eigenvalues and their associated eigenstates, but also integrates them into a unified spectral picture in which extended bands, possible flat bands, and defect-induced localized modes coexist. This structure explains the diverse localization and spreading phenomena observed in space-inhomogeneous three-state quantum walks and highlights the utility of spectral methods for their analysis.

    \subsection{Connections to topological phases and quantum algorithms}
    From a physical perspective, the exponentially localized eigenstates that appear
    at boundaries between different asymptotic coins can be naturally interpreted as
    edge states in the language of topological quantum walks~\cite{Kitagawa,Obuse2011}.
    In particular, in models where the bulk time-evolution operators admit a topological classification, the transfer-matrix method detects discrete eigenvalues in quasi-energy gaps as elements of $\sigma_{\mathrm{disc}}(U)$ in \cref{theorem:spectrum-decomposition}; these eigenvalues correspond to the gap states associated with bulk--edge correspondence. On a more fundamental level, it would be an interesting direction to explore the relationship between the spectral structure analyzed in this study and the index theorems \cite{Cedzich2016}. Specifically, clarifying whether the number of solutions to $\chi(\lambda)=0$ corresponds to a certain Fredholm index would provide a more rigorous mathematical foundation for the bulk--edge correspondence in three-state models. Moreover, the structure of the flat-band sector carries additional physical significance: gapped flat bands are relatively easy to isolate and can support strongly correlated or Mott phase, whereas gapless flat bands may give rise to critical or semimetal-like behavior and are associated with unconventional transport phenomena \cite{flat-band}. 
    A systematic derivation of the relevant bulk invariants for general three-state walks is left for future study.

    Quantum walks also provide a powerful framework for designing search algorithms on graphs and lattices, and more generally, preparing target states with high probability. A key insight is that the success probability and time-to-solution of quantum walk search depend critically
    on the spectral properties of the underlying walk operator. For instance, lackadaisical quantum
    walks \cite{Lackadaisical1,Lackadaisical2,Lackadaisical3,Lackadaisical4,Lackadaisical5} where additional self-loops are introduced at each vertex—have been widely studied as an effective mechanism to improve search efficiency on both regular lattices and more general graphs. \cite{Unified} also showed that by introducing an interpolated walk—a modification that effectively adds a self-loop like stay transition to the underlying chain—one can construct a general framework for quantum-walk search algorithms on arbitrary graphs. In these models, the presence of self-loops modifies the coin operator and hence the spectral gap, directly influencing the algorithmic speedup.

    The connection between quantum singular value transformation (QSVT)—a general framework for quantum algorithms and split-step quantum walks used to probe topological phases has recently been explored in \cite{QSVT-kiumi}. Because QSVT is fundamentally about shaping and controlling the spectrum of embedded operators within a unitary evolution, it is natural to speculate that similar spectral principles may underlie both algorithmic speedups and topological phenomena. In particular, introducing suitable symmetries or structured coins in our three-state model may enrich its spectral landscape in ways reminiscent of topological quantum walks. Establishing a unified viewpoint that bridges topological phases and quantum speed-up via the spectral analysis of quantum walks therefore presents an intriguing direction for future research.

\subsection{Limitation and future works}
While our analysis covers a broad class of space-inhomogeneous three-state quantum walks, it still has important limitations. In this work we restrict attention to walks that are asymptotically homogeneous on both sides (two-phase tails) with only finitely many defects in between. This setting already includes many standard models such as finite-defect and two-phase interfaces, but it does not encompass completely general inhomogeneous coins with arbitrarily long-range or disordered variations.

    A natural direction for future work is to extend the transfer-matrix formalism
    beyond the two-phase setting. One promising avenue is to treat walks periodic tails by directly applying techiniques in \cite{Kiumi2022-pd}. An even richer generalization is to incorporate random or quasi-periodic inhomogeneities,
    where the spectral properties can be governed by Lyapunov exponents and
    ergodic cocycle theory, following approaches used in studies of localization
    and fractality in inhomogeneous quantum walks~\cite{Shikano2010}, and in
    recent works on Anderson localization for quasi-periodic CMV matrices and
    quantum walks~\cite{wang2019anderson,Cedzich2023}.

    Moreover, it opens up possibilities for extending similar spectral analyses to multi-state QWs involving many particles \cite{multiparticle1,multiparticle2,multiparticle3}, as well as to higher-dimensional QWs \cite{two-dim1,two-dim2}. A future work could be to extend our framework to multi-particle, multi-state QWs and to higher-dimensional lattices. For a two-dimensional grid \cite{two-dim1} that remains translationally invariant along one axis but not the other, we can apply a Fourier transform in the invariant direction, reducing the problem to an effective four-state walk on a one-dimensional lattice with two self-loops. This reduction lets us deploy our transfer-matrix formalism unchanged, opening the door to rigorous, analytic treatment of two-dimensional QWs, precisely the regime where rich localization phenomena, non-trivial topological phases, and scalable quantum-algorithmic applications are expected to emerge.

    We expect that this research will provide a foundational basis for the rigorous spectral analysis of the broader and more complex quantum-walk models outlined above, ultimately contributing to a comprehensive mathematical theory of inhomogeneous quantum walks.

\noindent \textbf{Data availability:} The simulation code used in this work is available online at \cite{github}.

\section*{Acknowledgment}
This work is supported by JSPS KAKENHI, Grant Number JP22KJ1408, JST ASPIRE Japan, Grant Number JPMJAP2319 and JST PRESTO Japan, Grant Number JPMJPR25F1.
\appendix
\section{Proof of \cref{prop:condition}\label{app:proof_prop_condition}}
In this appendix, we provide a detailed derivation of the transfer matrix. The transfer matrix plays a key role in simplifying the study of the QW's eigenvalue equation \( U\Psi = e^{i\lambda} \Psi \), allowing us to reduce the problem to a \( 2 \times 2 \) matrix that encapsulates the relationships between components of the quantum state across different lattice sites. The transfer matrix \( T_x(\lambda) \) is defined such that it relates $\tilde{\Psi}\in\ell ^{2} (\mathbb{Z} ;\mathbb{C}^{2} )$ at adjacent positions \( x \) and \( x+1 \) through following conditions:

\begin{align}
     & e^{i\lambda } \tilde{\Psi} _{1} (x)=A_{x} (\lambda )\tilde{\Psi} _{1} (x+1)+B_{x} (\lambda )\tilde{\Psi} _{2} (x), \label{app:first}   \\
     & e^{i\lambda} \tilde{\Psi} _{2} (x+1)=C_{x} (\lambda )\tilde{\Psi} _{1} (x+1)+D_{x} (\lambda )\tilde{\Psi} _{2} (x). \label{app:second}
\end{align}
where
\begin{equation*}
    \begin{aligned}
         & A_{x}( \lambda ) =a_{x}^{( 1,1)} +\frac{a_{x}^{( 1,2)} a_{x}^{( 2,1)}}{e^{i\lambda } -a_{x}^{( 2,2)}} ,\ B_{x}( \lambda ) =a_{x}^{( 1,3)} +\frac{a_{x}^{( 1,2)} a_{x}^{( 2,3)}}{e^{i\lambda } -a_{x}^{( 2,2)}} , \\
         & C_{x}( \lambda ) =a_{x}^{( 3,1)} +\frac{a_{x}^{( 3,2)} a_{x}^{( 2,1)}}{e^{i\lambda } -a_{x}^{( 2,2)}} ,\ D_{x}( \lambda ) =a_{x}^{( 3,3)} +\frac{a_{x}^{( 3,2)} a_{x}^{( 2,3)}}{e^{i\lambda } -a_{x}^{( 2,2)}} ,
    \end{aligned}
\end{equation*}
so that $\iota^{-1}\tilde{\Psi}\in\mathcal{H}$ satisfies the eigenvalue equation \( U(\iota^{-1}\tilde{\Psi}) = e^{i\lambda} (\iota^{-1}\tilde{\Psi}) \). Then, the transfer matrix is derived by the reformulation of \cref{app:first,app:second} as
\begin{equation*}
    \tilde{\Psi}(x+1) =T_{x} (\lambda )\tilde{\Psi}(x),
\end{equation*}
where the transfer matrix $ T_{x} (\lambda )$ is given by
\begin{align*}
    \frac{1}{A_{x} (\lambda )}\left[\begin{array}{ c c }
                                            e^{i\lambda }    & -B_{x} (\lambda )                                                                   \\
                                            C_{x} (\lambda ) & -e^{-i\lambda }( B_{x} (\lambda )C_{x} (\lambda )-A_{x} (\lambda )D_{x} (\lambda ))
                                        \end{array}\right].
\end{align*}
From this expression, the transfer matrix can be constructed only if $A_{x} (\lambda )\neq 0$, and we will discuss this case later in this section. Recall that the coin matrix $\mathsf{C}_{x}$ is a $3\times 3$ unitary matrices, which is written as below:
\begin{equation*}
    \mathsf{C}_{x} =\left[\begin{array}{ c c c }
            a_{x}^{( 1,1)} & a_{x}^{( 1,2)} & a_{x}^{( 1,3)} \\
            a_{x}^{( 2,1)} & a_{x}^{( 2,2)} & a_{x}^{( 2,3)} \\
            a_{x}^{( 3,1)} & a_{x}^{( 3,2)} & a_{x}^{( 3,3)}
        \end{array}\right].
\end{equation*}
The unitarity of the coin matrix implies that the conjugate transpose is equal to the inverse of $C_x$. Therefore, we can write
\begin{equation*}
    \mathsf{C}_{x} =e^{i\Delta _x}\left[\begin{array}{ c c c }
            \overline{\det \mathsf{C}_{x}^{( 1,1)}}  & -\overline{\det \mathsf{C}_{x}^{( 1,2)}} & \overline{\det \mathsf{C}_{x}^{( 1,3)}}  \\
            -\overline{\det \mathsf{C}_{x}^{( 2,1)}} & \overline{\det \mathsf{C}_{x}^{( 2,2)}}  & -\overline{\det \mathsf{C}_{x}^{( 2,3)}} \\
            \overline{\det \mathsf{C}_{x}^{( 3,1)}}  & -\overline{\det \mathsf{C}_{x}^{( 3,2)}} & \overline{\det \mathsf{C}_{x}^{( 3,3)}}
        \end{array}\right] ,
\end{equation*}
where $e^{i\Delta _{x}} =\det \mathsf{C}_{x}$ and $\mathsf{C}_{x}^{( i,j)}$ denote the $2\times 2$ submatrix obtained by excluding the $i$-th row and $j$-th column of $\mathsf{C}_{x}$. Using the above expression, we can simplify each component of the transfer matrix as:
\centerline{$\displaystyle A_{x} (\lambda )=\frac{a_{x}^{( 1,1)} e^{i\lambda } -e^{i\Delta _{x}}\overline{a_{x}^{( 3,3)}}}{e^{i\lambda } -a_{x}^{( 2,2)}} ,$}
\centerline{$\displaystyle B_{x} (\lambda )=\frac{a_{x}^{( 1,3)} e^{i\lambda } +e^{i\Delta _{x}}\overline{a_{x}^{( 3,1)}}}{e^{i\lambda } -a_{x}^{( 2,2)}} ,$}
\centerline{$\displaystyle C_{x} (\lambda )=\frac{a_{x}^{( 3,1)} e^{i\lambda } +e^{i\Delta _{x}}\overline{a_{x}^{( 1,3)}}}{e^{i\lambda } -a_{x}^{( 2,2)}} ,$}
\centerline{$\displaystyle D_{x} (\lambda )=\frac{a_{x}^{( 3,3)} e^{i\lambda } -e^{i\Delta _{x}}\overline{a_{x}^{( 1,1)}}}{e^{i\lambda } -a_{x}^{( 2,2)}},$}
and we can further calculate
\begin{align*}
    A_{x} (\lambda )D_{x} & (\lambda )-B_{x} (\lambda )C_{x} (\lambda )                                                                                                                                                               \\
                          & =\frac{\det \mathsf{C}_{x}^{( 2,2)} e^{i\lambda } -e^{i\Delta _{x}}}{e^{i\lambda } -a_{x}^{( 2,2)}}=-e^{i( \lambda +\Delta _{x})}\frac{e^{-i\lambda } -\overline{a_{x}^{( 2,2)}}}{e^{i\lambda } -a_{x}^{( 2,2)}} .
\end{align*}
Hence, we can simplify the expression of the transfer matrix as
\begin{widetext}
    \begin{equation*}
        T_{x}( \lambda ) =\frac{1}{a_{x}^{( 1,1)} e^{i\lambda } -e^{i\Delta _{x}}\overline{a_{x}^{( 3,3)}}}\begin{bmatrix}
            e^{i\lambda }\left( e^{i\lambda } -a_{x}^{( 2,2)}\right)                & -a_{x}^{( 1,3)} e^{i\lambda } -e^{i\Delta _{x}}\overline{a_{x}^{( 3,1)}} \\
            a_{x}^{( 3,1)} e^{i\lambda } +e^{i\Delta _{x}}\overline{a_{x}^{( 1,3)}} & -e^{i\Delta _{x}}\left( e^{-i\lambda } -\overline{a_{x}^{( 2,2)}}\right)
        \end{bmatrix} .
    \end{equation*}
\end{widetext}
Nextly, we assume $A_{x} (\lambda )= 0$ where we cannot construct the transfer matrix. First we state the following lemma.
\begin{lemma}
    For all $\lambda \in ( 0,2\pi ]$, we have
    \begin{equation*}
        A_{x}( \lambda ) =0\iff D_{x}( \lambda ) =0.
    \end{equation*}
    \begin{proof}
        We assume that $A_{x}( \lambda ) =0$, which means
        \begin{equation*}
            e^{i( \lambda -\Delta _{x})} =\frac{\overline{a_{x}^{( 3,3)}}}{a_{x}^{( 1,1)}}\text{ and }\left| a_{x}^{( 3,3)}\right| =\left| a_{x}^{( 1,1)}\right| .
        \end{equation*}
        Then, we can easily calculate that
        \begin{equation*}
            D_{x}( \lambda ) =\frac{\left| a_{x}^{( 3,3)}\right| ^{2} -\left| a_{x}^{( 1,1)}\right| ^{2}}{\overline{a_{x}^{( 3,3)}} -a_{x}^{( 1,1)} a_{x}^{( 2,2)} e^{-i\Delta _{x}}} =0.
        \end{equation*}
        Similarly $D_{x}( \lambda ) =0$ implies
        \begin{equation*}
            e^{i( \lambda -\Delta _{x})} =\frac{\overline{a_{x}^{( 1,1)}}}{a_{x}^{( 3,3)}}\text{ and }\left| a_{x}^{( 1,1)}\right| =\left| a_{x}^{( 3,3)}\right|
        \end{equation*}
        and $A_{x}( \lambda ) =0$.

    \end{proof}
\end{lemma}
Assume $A_{x}( \lambda ) =0.$ By our previous lemma, this implies $D_{x}( \lambda ) =0.$ Given this, we have $e^{i\lambda } =e^{i\Delta _{x}}\frac{\overline{a_{x}^{( 3,3)}}}{a_{x}^{( 1,1)}} =e^{i\Delta _{x}}\frac{\overline{a_{x}^{( 1,1)}}}{a_{x}^{( 3,3)}}$ and

\begin{align*}
    \tilde{\Psi }_{1} (x)   & =e^{-i\lambda } B_{x}( \lambda )\tilde{\Psi }_{2} (x),       \\
    \tilde{\Psi }_{1} (x+1) & =e^{i\lambda } C_{x}^{-1}( \lambda )\tilde{\Psi }_{2} (x+1),
\end{align*}
where
\begin{gather*}
    B_{x}( \lambda ) =\frac{a_{x}^{( 1,3)}\overline{a_{x}^{( 3,3)}} +a_{x}^{( 1,1)}\overline{a_{x}^{( 3,1)}}}{\overline{a_{x}^{( 3,3)}} -e^{-i\Delta _{x}} a_{x}^{( 1,1)} a_{x}^{( 2,2)}} =e^{i\Delta _{x}}\frac{\overline{a_{x}^{( 3,2)}}}{a_{x}^{( 2,1)}} ,\ \\
    C_{x}( \lambda ) =\frac{a_{x}^{( 3,1)}\overline{a_{x}^{( 3,3)}} +a_{x}^{( 1,1)}\overline{a_{x}^{( 1,3)}}}{\overline{a_{x}^{( 3,3)}} -a_{x}^{( 1,1)} a_{x}^{( 2,2)} e^{-i\Delta _{x}}} =e^{i\Delta _{x}}\frac{\overline{a_{x}^{( 2,3)}}}{a_{x}^{( 1,2)}} .
\end{gather*}
Therefore, we can further calculate the equations:
\begin{align*}
    \overline{a_{x}^{( 1,1)}} a_{x}^{( 2,1)}\tilde{\Psi }_{1} (x)  & =a_{x}^{( 3,3)}\overline{a_{x}^{( 3,2)}}\tilde{\Psi }_{2} (x),    \\
    a_{x}^{( 3,3)}\overline{a_{x}^{( 2,3)}}\tilde{\Psi }_{1} (x+1) & =\overline{a_{x}^{( 1,1)}} a_{x}^{( 1,2)}\tilde{\Psi }_{2} (x+1),
\end{align*}
and the statement is proved.

\section{Proof of Main results\label{app:proof1}}
This appendix is dedicated to the proof of the main results presented in this paper.

\cref{main theorem} and its corollaries, which establish the necessary and sufficient conditions for the eigenvalue problem, and they can be directly proved by following the discussion in \cref{subsec:main-theorem} to consider the square summability of $\tilde{\Psi}$ in \cref{prop:condition}. \cref{prop:det} and \cref{prop:necessary} are the core properties of the transfer matrix, which are essential for the proof of the main theorem. The proof of these propositions is provided in \cref{app:det} and \cref{app:necessary}, respectively.

\subsection{Proof of \cref{prop:det}\label{app:det}}
\begin{proof}
    By direct calculation, we have
    \begin{align*}
         & \left| a_{x}^{( 1,1)} e^{i\lambda } -e^{i\Delta _{x}}\overline{a_{x}^{( 3,3)}}\right| ^{2} =\left| a_{x}^{( 3,3)} e^{i\lambda } -e^{i\Delta _{x}}\overline{a_{x}^{( 1,1)}}\right| ^{2} \\
         & =\left| a_{x}^{( 1,1)}\right| ^{2} +\left| a_{x}^{( 3,3)}\right| ^{2} -2\Re \left( a_{x}^{( 1,1)} a_{x}^{( 3,3)} e^{i( \lambda -\Delta _{x})}\right).
    \end{align*}
    Therefore,
    \begin{equation*}
        \left| \frac{A_x(\lambda )}{D_x(\lambda )}\right| ^{2} =\frac{\left| a_{x}^{( 1,1)} e^{i\lambda } -e^{i\Delta _{x}}\overline{a_{x}^{( 3,3)}}\right| ^{2}}{\left| a_{x}^{( 3,3)} e^{i\lambda } -e^{i\Delta _{x}}\overline{a_{x}^{( 1,1)}}\right| ^{2}} =1
    \end{equation*}which implies
    \begin{equation*}
        | \det T_x(\lambda )| =\left| \frac{A_x(\lambda )}{D_x(\lambda )}\right| =1.
    \end{equation*}
\end{proof}
\subsection{Proof of \cref{prop:necessary}\label{app:necessary}}
\begin{proof}
    To begin, we establish that
    \begin{equation}\label{eq:b1}
        \mathrm{tr}(T_{\pm \infty}) = \det(T_{\pm \infty}) \overline{\mathrm{tr}(T_{\pm \infty})}.
    \end{equation}

    The trace and determinant of the transfer matrix can be calculated as:
    \begin{align*}
         & \mathrm{tr}(T_{\pm \infty}) = \frac{e^{i \lambda} (e^{i \lambda} - a_{x}^{(2,2)}) - e^{i \Delta_x} (e^{-i \lambda} - \overline{a_{x}^{(2,2)}})}{a_{x}^{(1,1)} e^{i \lambda} - e^{i \Delta_x} \overline{a_{x}^{(3,3)}}},
        \\
         & \det(T_{x}(\lambda)) = \frac{D_x(\lambda)}{A_x(\lambda)}.
    \end{align*}
    The trace can be rewritten using \( A_x(\lambda) \) and \( D_x(\lambda) \) as:
    \begin{equation*}
        \mathrm{tr}(T_{\pm \infty}) = \frac{e^{i \lambda}}{A_x(\lambda)} + \frac{e^{-i \lambda}}{\overline{D_x(\lambda)}}.
    \end{equation*}

    Since \( |A_x(\lambda)| = |D_x(\lambda)| \) by a previous proposition, it follows directly that \( \mathrm{tr}(T_{\pm \infty}) = \det(T_{\pm \infty}) \overline{\mathrm{tr}(T_{\pm \infty})} \) holds.

    Now, we proceed to the main statement. To find the necessary and sufficient condition for \( \left|\zeta_{\pm \infty}^+\right| \neq \left|\zeta_{\pm \infty}^-\right| \), we consider the case where
    \begin{equation*}
        \left|\zeta_{\pm \infty}^+\right| = \left|\zeta_{\pm \infty}^-\right|,
    \end{equation*}
    which is equivalent to
    \begin{equation}\label{eq:b2}
        \Re\left(\overline{\mathrm{tr}(T_{\pm \infty})} \sqrt{\mathrm{tr}(T_{\pm \infty})^2 - 4 \det(T_{\pm \infty})}\right) = 0.
    \end{equation}

    Assuming \( \mathrm{tr}(T_{\pm \infty}) \neq 0 \), by \cref{eq:b1}, we have
    \begin{equation*}
        \det(T_{\pm \infty}) = \frac{\mathrm{tr}(T_{\pm \infty})}{\overline{\mathrm{tr}(T_{\pm \infty})}}.
    \end{equation*}
    Substituting this into \cref{eq:b2} gives
    \begin{align}
        \overline{\mathrm{tr}(T_{\pm \infty})}
         & \sqrt{\mathrm{tr}(T_{\pm \infty})^2 - 4 \det(T_{\pm \infty})}\nonumber
        \\
         & =\sqrt{|\mathrm{tr}(T_{\pm \infty})|^2 (|\mathrm{tr}(T_{\pm \infty})|^2 - 4)}.
    \end{align}

    Thus, the condition \( |\mathrm{tr}(T_{\pm \infty})|^2 - 4 \leq 0 \) is necessary and sufficient for \cref{eq:b2} to hold, meaning \( \left|\zeta_{\pm \infty}^+\right| = \left|\zeta_{\pm \infty}^-\right| \). If \( \mathrm{tr}(T_{\pm \infty}) = 0 \), it is clear that \( \left|\zeta_{\pm \infty}^+\right| = \left|\zeta_{\pm \infty}^-\right| \) holds as well, with \( |\mathrm{tr}(T_{\pm \infty})|^2 - 4 = -4 \leq 0 \).

    By concluding the above discussions, we have proved \( \left|\zeta_{\pm \infty}^+\right| \neq \left|\zeta_{\pm \infty}^-\right| \) if and only if \( |\mathrm{tr}(T_{\pm \infty})|^2 - 4 > 0 \).
\end{proof}

\end{document}